\documentclass[11pt, journal]{IEEEtran}

\usepackage{amsmath}
\usepackage{amsthm}
\usepackage{color}
\usepackage{graphicx}
\usepackage[usenames,dvipsnames]{pstricks}
\usepackage{epsfig}
\usepackage{pst-grad} 
\usepackage{pst-plot} 
\usepackage{upgreek} 
\usepackage{flushend}

\newtheorem{proposition}{Proposition}

\newtheorem{lemma}{Lemma}

\def\B{L}

\def\G{R}

\def\cW{\mathcal{W}}

\def\cY{\mathcal{Y}}
\def\cZ{\mathcal{Z}}

\author{
\centerline{Eren \c Sa\c so\u glu\quad and\quad Lele Wang}

\thanks{E. \c Sa\c so\u glu is with the Department of Electrical
Engineering and Computer Sciences, University of California, Berkeley,
CA 94720 USA (email: eren@eecs.berkeley.edu).}
\thanks{L. Wang is with the Department of Electrical and Computer
Engineering, University of California, San Diego, La Jolla, CA 92093
USA (email: lew001@ucsd.edu).}
}
\title{Universal Polarization}
\begin{document}
\maketitle
\begin{abstract}
A method to polarize channels universally is introduced.  The method
is based on combining two distinct channels in each polarization step
as opposed to Ar\i kan's original method of combining identical
channels.  This creates an equal number of only two types of channels,
one of which becomes progressively better as the other becomes worse.
The locations of the good polarized channels are independent of the
underlying channel, guaranteeing universality.  Polarizing the good
channels further with Ar\i kan's method results in universal polar
codes of rate~$1/2$.  The method is generalized to construct codes of
arbitrary rates.

It is also shown that the \emph{less noisy} ordering of channels is
preserved under polarization, and thus a good polar code for a given
channel will perform well over a less noisy one.

\end{abstract}

\section{Introduction}

The compound channel models communication without perfect knowledge of
the physical channel.  The channel is assumed to belong to a certain
class, and a code needs to be designed to perform well over all
members of the class.   The problem is relevant from a practical
standpoint since one can rarely estimate the channel perfectly, and it
is undesirable for small variations in the channel to change the code
performance dramatically.

Let $\cW$ be a class of binary-input memoryless channels, and
let~$I(\cW)$ denote the symmetric compound capacity of~$\cW$, which is
the highest achievable rate over all $W\in\cW$ by codes with an equal
frequency of zeros and ones.  It is known~\cite{BBT1959} that
$$
I(\cW)=\inf_{W\in\cW} I(W)
$$
where~$I(W)$ denotes the symmetric-capacity of~$W$.  A code sequence
of rate~$R$ is said to be \emph{universal} if its error probability
vanishes over the class of all channels with $I(W)>R$.

A problem of practical interest is to design universal codes with low
encoding and decoding complexities.  To this end, Kudekar \emph{et.\
al.}~\cite{KRU2012} recently showed that spatially-coupled LDPC codes
are universal (for symmetric channels) under low-complexity
message-passing decoders.  Here, we investigate whether universality
can be attained by Ar\i kan's polarization methods~\cite{Arikan2009}.
As in~\cite{KRU2012}, we consider the setting where the channel is
unknown only to the transmitter.  This is an idealized version of the
practical scenario where the receiver may estimate the channel prior
to data transmission, e.g., through the use of training symbols.
Polar coding for this setting was first considered by Hassani
\emph{et.\ al.}~\cite{HKU2009}, who concluded that Ar\i kan's original
codes are not universal under successive cancellation (SC) decoding.
It is worth noting, however, that polar codes are universal under the
optimal but computationally unfeasible maximum likelihood
decoding~\cite[pp.~87--89]{Sasoglu2011}.  


There are cases in which designing a polar code for multiple channels
is easy.  The most prominent of these is the \emph{degraded} case: A
polar code tailor to a given channel will also perform well over all
upgraded versions of that channel~\cite{Korada2009}.  In Appendix~A,
we show that a similar statement holds for the more general class of
\emph{less noisy} comparable channels.  

Our aim here is to show a method to polarize channels universally.  We
will first discuss how to achieve rate $1/2$, and later show
constructions that achieve arbitrary rates.

\section{Method} 
\label{sec:method}

As in Ar{\i}kan's original method, we will polarize channels
recursively.  The construction will have two stages, which we call the
\emph{slow} polarization and the \emph{fast} polarization stages.
Slow polarization will create only two types of channels after each
recursion:  Almost half of the polarized channels will be of the first
type and become increasingly good, the other half will become
increasingly bad.  The indices of the good channels will be
independent of the underlying channel, and thus universality will be
attained at this stage.  However, this type of polarization will be
too slow to allow reliable SC decoding.  In order to improve
reliability, we will switch to Ar{\i}kan's regular (fast) polarization
method once sufficient universality is achieved. 

Given two binary-input memoryless channels $W\colon\{0,1\}\to\cY$ and
$V\colon\{0,1\}\to\cZ$, define the binary-input channels
$$
(W,V)^-(y,z\mid x)
	=\sum_{u\in\{0,1\}} 
	\tfrac12 W(y\mid u+x)V(z\mid u)
$$
and
$$
(W,V)^+(y,z,u\mid x)
	=\tfrac12 W(y\mid u+x)V(z\mid x).
$$
Note that if $W=V$, then these are equivalent to the regular polarized
channels $W^-$ and $W^+$ in~\cite{Arikan2009}.  We will let $\B_n$ and
$\G_n$ denote the two channels that will emerge in the $n$th level of
slow polarization.  These are defined recursively through 
\begin{align}
\notag
\B_0&=\G_0=W\\
\label{eqn:slow-def}
\begin{split}
\B_{n+1}&=(\B_n,\G_n)^-\\
\G_{n+1}&=(\B_n,\G_n)^+
\end{split}
\qquad n=0,1,\dotsc
\end{align}
Observe that each recursion except the first combines two different
channels to produce the channels of the next level.  This is in
contrast with the original polarization method, which combines
identical channels to create $2^n$ polarized channels at the $n$th
recursion,
\begin{align*}
\begin{split}
W^{\bf s-}&=(W^{\bf s},W^{\bf s})^-\\
W^{\bf s+}&=(W^{\bf s},W^{\bf s})^+
\end{split},
\quad {\bf s}\in\{-,+\}^{n-1}.
\end{align*}
It is readily seen that for all $n$ we have
$$
I(\B_n)+I(\G_n)=2I(W).
$$ 
Standard arguments also show that $I(\B_n)$ is decreasing and
$I(\G_n)$ is increasing:
\begin{align*}
I(\B_{n+1})\le I(\B_n)\le I(\G_n)\le I(\G_{n+1}).
\end{align*}
Since both $I(\B_n)$ and $I(\G_n)$ are monotone and bounded by
$0$~and~$1$, they have $[0,1]$-valued limits, which we respectively
call $I(\B_\infty)$ and $I(\G_\infty)$.  Further, it follows
from~\cite[Lemma~2.1]{Sasoglu2011} that the inequalities above are
strict unless $I(\B_n)\in\{0,1\}$ or $I(\G_n)\in\{0,1\}$.   This
implies the following polarization result.

\begin{proposition}
\label{prop:slow}
\hfill\phantom{x}\\
\vspace{-2ex}
\begin{itemize}
\item[(i)]
If $I(W)\ge 1/2$, then 
\begin{align*}
I(\B_\infty)=2I(W)-1,
\quad
I(\G_\infty)=1.
\end{align*}
\item[(ii)]
If $I(W)\le 1/2$, then 
\begin{align*}
I(\B_\infty)=0,
\quad
I(\G_\infty)=2I(W).
\end{align*}
\end{itemize}
\end{proposition}
We now describe a transform that recursively produces the channels
$\B_n$~and~$\G_n$.  This is best done graphically; the claims will be
evident from the figures.  Note first that $\B_1$~and~$\G_1$ are
identical to $W^-$~and~$W^+$, and thus can be obtained in the regular
manner (Figure~\ref{fig:basic}).  In order to create $\B_2$~and~$\G_2$
from these, one can take two independent $(\B_1,\G_1)$ pairs, and
combine the $\B_1$ from one pair with the $\G_1$ from the other, as in
Figure~\ref{fig:2-short}.  Inspecting the figure, one may be tempted
to combine the second $\B_1$~and~$\G_1$ to obtain another
$(\B_2,\G_2)$ pair, but some thought reveals that doing so fails to create the
desired effect.  

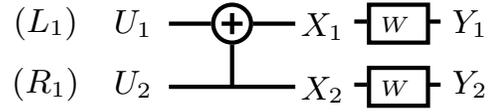
\begin{figure}[t]
\centerline{
\scalebox{1.4} 
{
\begin{pspicture}(-1.2,-0.5525)(4.6090627,0.5459375)
\psline[linewidth=0.04cm](0.87,0.3125)(1.27,0.3125)
\pscircle[linewidth=0.04,dimen=outer](1.47,0.3125){0.2}
\psline[linewidth=0.04cm](1.67,0.3125)(2.07,0.3125)
\psline[linewidth=0.04cm](1.47,0.4125)(1.47,0.2125)
\psline[linewidth=0.04cm](1.37,0.3125)(1.57,0.3125)
\psline[linewidth=0.04cm](1.47,0.1125)(1.47,-0.2875)
\psline[linewidth=0.04cm](0.87,-0.2875)(2.07,-0.2875)
\psframe[linewidth=0.04,dimen=outer](3.35,0.5125)(2.75,0.1125)
\psline[linewidth=0.04cm](2.63,0.3325)(2.75,0.3325)
\psline[linewidth=0.04cm](3.33,0.3325)(3.45,0.3325)
\psframe[linewidth=0.04,dimen=outer](3.35,-0.0875)(2.75,-0.4875)
\psline[linewidth=0.04cm](2.63,-0.2675)(2.75,-0.2675)
\psline[linewidth=0.04cm](3.33,-0.2675)(3.45,-0.2675)
{\footnotesize
\usefont{T1}{ptm}{m}{n}
\rput(3.73,0.3125){$Y_1$}
\rput(3.73,-0.2675){$Y_2$}
\rput(2.33,0.2825){$X_1$}
\rput(2.33,-0.3175){$X_2$}
\rput(0.53,0.3125){$U_1$}
\rput(0.53,-0.2675){$U_2$}
\rput(-0.3,0.3125){$(\B_1)$}
\rput(-0.3,-0.2675){$(\G_1)$}
}
{\tiny
\usefont{T1}{ptm}{m}{n}
\rput(3.02,0.3025){$W$}
\rput(3.02,-0.2975){$W$}
}
\end{pspicture} 
}}
\caption{
\label{fig:basic}
When $X_1$ and $X_2$ are uniform and independent, the channel $U_1\to
Y_1^2$ is equivalent to $\B_1=W^-$ and $U_2\to Y_1^2U_1$ is equivalent to
$\G_1=W^+$.}
\end{figure}
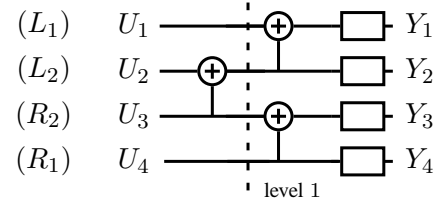
\begin{figure}[t]
\centerline{
\scalebox{1} 
{
\begin{pspicture}(0,-1.3664062)(6.519062,1.3464062)
\psframe[linewidth=0.04,dimen=outer](5.485469,1.2129687)(4.885469,0.81296873)
\psline[linewidth=0.04cm](4.7654686,1.0129688)(4.885469,1.0129688)
\psline[linewidth=0.04cm](5.465469,1.0129688)(5.585469,1.0129688)
\psframe[linewidth=0.04,dimen=outer](5.485469,0.61296874)(4.885469,0.21296875)
\psline[linewidth=0.04cm](4.7654686,0.41296875)(4.885469,0.41296875)
\psline[linewidth=0.04cm](5.465469,0.41296875)(5.585469,0.41296875)
\psline[linewidth=0.04cm](3.3654687,-0.18703125)(3.8654687,-0.18703125)
\pscircle[linewidth=0.04,dimen=outer](4.065469,-0.18703125){0.2}
\psline[linewidth=0.04cm](4.2654686,-0.18703125)(4.7654686,-0.18703125)
\psline[linewidth=0.04cm](4.065469,-0.08703125)(4.065469,-0.28703126)
\psline[linewidth=0.04cm](3.9654686,-0.18703125)(4.1654687,-0.18703125)
\psline[linewidth=0.04cm](4.065469,-0.38703126)(4.065469,-0.78703123)
\psline[linewidth=0.04cm](3.3854687,-0.78703123)(4.8654685,-0.78703123)
\psline[linewidth=0.04cm](2.4854686,0.41296875)(2.9854686,0.41296875)
\pscircle[linewidth=0.04,dimen=outer](3.1854687,0.41296875){0.2}
\psline[linewidth=0.04cm](3.3854687,0.41296875)(3.8854687,0.41296875)
\psline[linewidth=0.04cm](3.1854687,0.5129688)(3.1854687,0.31296876)
\psline[linewidth=0.04cm](3.0854688,0.41296875)(3.2854688,0.41296875)
\psline[linewidth=0.04cm](3.1854687,0.21296875)(3.1854687,-0.18703125)
\psline[linewidth=0.04cm](2.4854686,-0.18703125)(3.8854687,-0.18703125)
\psline[linewidth=0.04cm](3.3654687,1.0129688)(3.8654687,1.0129688)
\pscircle[linewidth=0.04,dimen=outer](4.065469,1.0129688){0.2}
\psline[linewidth=0.04cm](4.2654686,1.0129688)(4.7654686,1.0129688)
\psline[linewidth=0.04cm](4.065469,1.1129688)(4.065469,0.91296875)
\psline[linewidth=0.04cm](3.9654686,1.0129688)(4.1654687,1.0129688)
\psline[linewidth=0.04cm](4.065469,0.81296873)(4.065469,0.41296875)
\psline[linewidth=0.04cm](3.3654687,0.41296875)(4.7654686,0.41296875)
\psline[linewidth=0.04cm](2.4854686,1.0129688)(3.8854687,1.0129688)
\psframe[linewidth=0.04,dimen=outer](5.485469,0.01296875)(4.885469,-0.38703126)
\psline[linewidth=0.04cm](4.7654686,-0.18703125)(4.885469,-0.18703125)
\psline[linewidth=0.04cm](5.465469,-0.18703125)(5.585469,-0.18703125)
\psframe[linewidth=0.04,dimen=outer](5.485469,-0.58703125)(4.885469,-0.9870312)
\psline[linewidth=0.04cm](4.7654686,-0.78703123)(4.885469,-0.78703123)
\psline[linewidth=0.04cm](5.465469,-0.78703123)(5.585469,-0.78703123)
\usefont{T1}{ppl}{m}{n}
\rput(5.944531,1.0229688){$Y_1$}
\usefont{T1}{ppl}{m}{n}
\rput(5.944531,0.42296875){$Y_2$}
\usefont{T1}{ppl}{m}{n}
\rput(5.944531,-0.17703125){$Y_3$}
\usefont{T1}{ppl}{m}{n}
\rput(5.944531,-0.77703124){$Y_4$}
\usefont{T1}{ppl}{m}{n}
\rput(2.1445312,1.0229688){$U_1$}
\usefont{T1}{ppl}{m}{n}
\rput(2.1445312,0.42296875){$U_2$}
\usefont{T1}{ppl}{m}{n}
\rput(2.1445312,-0.17703125){$U_3$}
\usefont{T1}{ppl}{m}{n}
\rput(2.1445312,-0.77703124){$U_4$}
\usefont{T1}{ppl}{m}{n}
\rput(0.94453126,1.0229688){$(\B_1)$}
\usefont{T1}{ppl}{m}{n}
\rput(0.94453126,0.42296875){$(\B_2)$}
\usefont{T1}{ppl}{m}{n}
\rput(0.94453126,-0.17703125){$(\G_2)$}
\usefont{T1}{ppl}{m}{n}
\rput(0.94453126,-0.77703124){$(\G_1)$}
\psline[linewidth=0.04cm](2.5454688,-0.78703123)(3.9454687,-0.78703123)
\psline[linewidth=0.04cm,linestyle=dashed,dash=0.10cm 0.16cm](3.66,1.3264062)(3.66,-1.1735938)
{\scriptsize
\usefont{T1}{ptm}{m}{n}
\rput(4.2525,-1.1635938){level $1$}
}
\end{pspicture} 
}}
\caption{
\label{fig:2-short}
The channels $U_i\to Y_1^4U_1^{i-1}$ are equivalent to those inside
the parentheses.
}
\end{figure}

Instead, more $(\B_2,\G_2)$ pairs can be obtained by combining more
than two $(\B_1,\G_1)$ pairs in a chain.  This is shown in
Figure~\ref{fig:2-level},  where four $(\B_1,\G_1)$ pairs are chained.
The resulting transform creates three~$(\B_2,\G_2)$ pairs.  One can
more generally chain~$K$ $(\B_1,\G_1)$ pairs to produce~$K-1$
$(\B_2,\G_2)$ pairs.  Thus, the fraction of $(\B_2,\G_2)$ pairs can be
made as close to~$1$ as desired by taking $K$ sufficiently large.
Observe also that the channels $U_i\to Y_1^{2K}U_1^{i-1}$ obtained by
such a chain are equivalent to $U_i\to Y_{i-2}^{i+2}U_{i-2}^{i-1}$.
That is, not all channel outputs are relevant to~$U_i$. 

\begin{figure}[t]
\centerline{
\scalebox{1} 
{
\begin{pspicture}(0,-2.7664063)(6.519062,2.7464063)
\psframe[linewidth=0.04,dimen=outer](5.5109377,2.6129687)(4.910938,2.2129688)
\psline[linewidth=0.04cm](4.7909374,2.4129686)(4.910938,2.4129686)
\psline[linewidth=0.04cm](5.4909377,2.4129686)(5.6109376,2.4129686)
\psframe[linewidth=0.04,dimen=outer](5.5109377,2.0129688)(4.910938,1.6129688)
\psline[linewidth=0.04cm](4.7909374,1.8129687)(4.910938,1.8129687)
\psline[linewidth=0.04cm](5.4909377,1.8129687)(5.6109376,1.8129687)
\psline[linewidth=0.04cm](2.5109375,-0.58703125)(3.0109375,-0.58703125)
\pscircle[linewidth=0.04,dimen=outer](3.2109375,-0.58703125){0.2}
\psline[linewidth=0.04cm](3.4109375,-0.58703125)(3.9109375,-0.58703125)
\psline[linewidth=0.04cm](3.2109375,-0.48703125)(3.2109375,-0.68703127)
\psline[linewidth=0.04cm](3.1109376,-0.58703125)(3.3109376,-0.58703125)
\psline[linewidth=0.04cm](3.2109375,-0.78703123)(3.2109375,-1.1870313)
\psline[linewidth=0.04cm](2.5109375,-1.1870313)(3.9109375,-1.1870313)
\psline[linewidth=0.04cm](3.3909373,0.01296875)(3.8909373,0.01296875)
\pscircle[linewidth=0.04,dimen=outer](4.0909376,0.01296875){0.2}
\psline[linewidth=0.04cm](4.2909374,0.01296875)(4.7909374,0.01296875)
\psline[linewidth=0.04cm](4.0909376,0.11296875)(4.0909376,-0.08703125)
\psline[linewidth=0.04cm](3.9909372,0.01296875)(4.1909375,0.01296875)
\psline[linewidth=0.04cm](4.0909376,-0.18703125)(4.0909376,-0.58703125)
\psline[linewidth=0.04cm](3.3909373,-0.58703125)(4.7909374,-0.58703125)
\psline[linewidth=0.04cm](3.3909373,1.2129687)(3.8909373,1.2129687)
\pscircle[linewidth=0.04,dimen=outer](4.0909376,1.2129687){0.2}
\psline[linewidth=0.04cm](4.2909374,1.2129687)(4.7909374,1.2129687)
\psline[linewidth=0.04cm](4.0909376,1.3129687)(4.0909376,1.1129688)
\psline[linewidth=0.04cm](3.9909372,1.2129687)(4.1909375,1.2129687)
\psline[linewidth=0.04cm](4.0909376,1.0129688)(4.0909376,0.61296874)
\psline[linewidth=0.04cm](3.4109375,0.61296874)(4.8909373,0.61296874)
\psline[linewidth=0.04cm](3.3909373,-1.1870313)(3.8909373,-1.1870313)
\pscircle[linewidth=0.04,dimen=outer](4.0909376,-1.1870313){0.2}
\psline[linewidth=0.04cm](4.2909374,-1.1870313)(4.7909374,-1.1870313)
\psline[linewidth=0.04cm](4.0909376,-1.0870312)(4.0909376,-1.2870313)
\psline[linewidth=0.04cm](3.9909372,-1.1870313)(4.1909375,-1.1870313)
\psline[linewidth=0.04cm](4.0909376,-1.3870312)(4.0909376,-1.7870313)
\psline[linewidth=0.04cm](3.3909373,-1.7870313)(4.7909374,-1.7870313)
\psline[linewidth=0.04cm](2.5109375,0.61296874)(3.0109375,0.61296874)
\pscircle[linewidth=0.04,dimen=outer](3.2109375,0.61296874){0.2}
\psline[linewidth=0.04cm](3.2109375,0.71296877)(3.2109375,0.5129688)
\psline[linewidth=0.04cm](3.1109376,0.61296874)(3.3109376,0.61296874)
\psline[linewidth=0.04cm](3.2109375,0.41296875)(3.2109375,0.01296875)
\psline[linewidth=0.04cm](2.5109375,0.01296875)(3.9109375,0.01296875)
\psline[linewidth=0.04cm](2.5109375,1.8129687)(3.0109375,1.8129687)
\pscircle[linewidth=0.04,dimen=outer](3.2109375,1.8129687){0.2}
\psline[linewidth=0.04cm](3.4109375,1.8129687)(3.9109375,1.8129687)
\psline[linewidth=0.04cm](3.2109375,1.9129688)(3.2109375,1.7129687)
\psline[linewidth=0.04cm](3.1109376,1.8129687)(3.3109376,1.8129687)
\psline[linewidth=0.04cm](3.2109375,1.6129688)(3.2109375,1.2129687)
\psline[linewidth=0.04cm](2.5109375,1.2129687)(3.9109375,1.2129687)
\psline[linewidth=0.04cm](3.3909373,2.4129686)(3.8909373,2.4129686)
\pscircle[linewidth=0.04,dimen=outer](4.0909376,2.4129686){0.2}
\psline[linewidth=0.04cm](4.2909374,2.4129686)(4.7909374,2.4129686)
\psline[linewidth=0.04cm](4.0909376,2.5129688)(4.0909376,2.3129687)
\psline[linewidth=0.04cm](3.9909372,2.4129686)(4.1909375,2.4129686)
\psline[linewidth=0.04cm](4.0909376,2.2129688)(4.0909376,1.8129687)
\psline[linewidth=0.04cm](3.3909373,1.8129687)(4.7909374,1.8129687)
\psline[linewidth=0.04cm](2.5109375,2.4129686)(3.9109375,2.4129686)
\psline[linewidth=0.04cm](3.9109375,-1.7870313)(2.5109375,-1.7870313)
\psframe[linewidth=0.04,dimen=outer](5.5109377,1.4129688)(4.910938,1.0129688)
\psline[linewidth=0.04cm](4.7909374,1.2129687)(4.910938,1.2129687)
\psline[linewidth=0.04cm](5.4909377,1.2129687)(5.6109376,1.2129687)
\psframe[linewidth=0.04,dimen=outer](5.5109377,0.81296873)(4.910938,0.41296875)
\psline[linewidth=0.04cm](4.7909374,0.61296874)(4.910938,0.61296874)
\psline[linewidth=0.04cm](5.4909377,0.61296874)(5.6109376,0.61296874)
\psframe[linewidth=0.04,dimen=outer](5.5109377,0.21296875)(4.910938,-0.18703125)
\psline[linewidth=0.04cm](4.7909374,0.01296875)(4.910938,0.01296875)
\psline[linewidth=0.04cm](5.4909377,0.01296875)(5.6109376,0.01296875)
\psframe[linewidth=0.04,dimen=outer](5.5109377,-0.38703126)(4.910938,-0.78703123)
\psline[linewidth=0.04cm](4.7909374,-0.58703125)(4.910938,-0.58703125)
\psline[linewidth=0.04cm](5.4909377,-0.58703125)(5.6109376,-0.58703125)
\psframe[linewidth=0.04,dimen=outer](5.5109377,-0.9870312)(4.910938,-1.3870312)
\psline[linewidth=0.04cm](4.7909374,-1.1870313)(4.910938,-1.1870313)
\psline[linewidth=0.04cm](5.4909377,-1.1870313)(5.6109376,-1.1870313)
\psframe[linewidth=0.04,dimen=outer](5.5109377,-1.5870312)(4.910938,-1.9870312)
\psline[linewidth=0.04cm](4.7909374,-1.7870313)(4.910938,-1.7870313)
\psline[linewidth=0.04cm](5.4909377,-1.7870313)(5.6109376,-1.7870313)
\usefont{T1}{ppl}{m}{n}
\rput(5.944531,2.4229689){$Y_1$}
\usefont{T1}{ppl}{m}{n}
\rput(5.944531,1.8229687){$Y_2$}
\usefont{T1}{ppl}{m}{n}
\rput(5.944531,1.2229687){$Y_3$}
\usefont{T1}{ppl}{m}{n}
\rput(5.944531,0.62296873){$Y_4$}
\usefont{T1}{ppl}{m}{n}
\rput(5.944531,0.02296875){$Y_5$}
\usefont{T1}{ppl}{m}{n}
\rput(5.944531,-0.57703125){$Y_6$}
\usefont{T1}{ppl}{m}{n}
\rput(5.944531,-1.1770313){$Y_7$}
\usefont{T1}{ppl}{m}{n}
\rput(5.944531,-1.7770313){$Y_8$}
\usefont{T1}{ppl}{m}{n}
\rput(2.1445312,2.4229689){$U_1$}
\usefont{T1}{ppl}{m}{n}
\rput(2.1445312,1.8229687){$U_2$}
\usefont{T1}{ppl}{m}{n}
\rput(2.1445312,1.2229687){$U_3$}
\usefont{T1}{ppl}{m}{n}
\rput(2.1445312,0.62296873){$U_4$}
\usefont{T1}{ppl}{m}{n}
\rput(2.1445312,0.02296875){$U_5$}
\usefont{T1}{ppl}{m}{n}
\rput(2.1445312,-0.57703125){$U_6$}
\usefont{T1}{ppl}{m}{n}
\rput(2.1445312,-1.1770313){$U_7$}
\usefont{T1}{ppl}{m}{n}
\rput(2.1445312,-1.7770313){$U_8$}
\usefont{T1}{ppl}{m}{n}
\rput(0.94453126,2.4229689){$(\B_1)$}
\usefont{T1}{ppl}{m}{n}
\rput(0.94453126,1.8229687){$(\B_2)$}
\usefont{T1}{ppl}{m}{n}
\rput(0.94453126,1.2229687){$(\G_2)$}
\usefont{T1}{ppl}{m}{n}
\rput(0.94453126,0.62296873){$(\B_2)$}
\usefont{T1}{ppl}{m}{n}
\rput(0.94453126,0.02296875){$(\G_2)$}
\usefont{T1}{ppl}{m}{n}
\rput(0.94453126,-0.57703125){$(\B_2)$}
\usefont{T1}{ppl}{m}{n}
\rput(0.94453126,-1.1770313){$(\G_2)$}
\usefont{T1}{ppl}{m}{n}
\rput(0.94453126,-1.7770313){$(\G_1)$}
\psline[linewidth=0.04cm,linestyle=dashed,dash=0.10cm 0.16cm](3.66,2.7264063)(3.66,-2.2735937)
{\scriptsize
\usefont{T1}{ptm}{m}{n}
\rput(4.2525,-2.2635939){level $1$}
}
\end{pspicture} 
}}
\caption{
\label{fig:2-level}
Four pairs of level-$1$ channels are chained to create six level-$2$
and two level-$1$ channels.  The channels $U_i\to Y_1^8U_1^{i-1}$ are
equivalent to the ones on the left.
}
\end{figure}
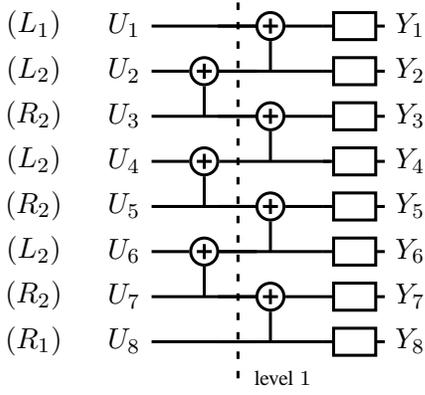

There are several ways to continue this construction to attain further
levels of polarization.  We describe here perhaps the simplest one,
where chaining as in Figure~\ref{fig:2-level} is used only in the
second recursion.  Each subsequent recursion combines only two blocks.
The third level of this construction with $K=4$ is shown in
Figure~\ref{fig:3-level}.  Here, only the level-$2$ channels $\B_2$
and $\G_2$ are combined in the third recursion, $\B_1$ and $G_1$ are
not.  Further, the first $\B_2$ in the first block and the last $\G_2$
in the second are also left unconnected, in order to ensure that the
remaining channels polarize to the third level to produce $\B_3$ and
$\G_3$.  This idea is easily extended to further levels: To obtain
$\B_{n+1}$ and $\G_{n+1}$ in the $(n+1)$-th recursion, one only
combines the $\B_n$s from the first block with the $\G_n$s from the
second, and vice versa.  The first $\B_n$ from the first block and the
last $\G_n$ from the second block are left unconnected.  This is shown
in Figure~\ref{fig:n-level}.

\begin{figure}[t]
\centerline{\input{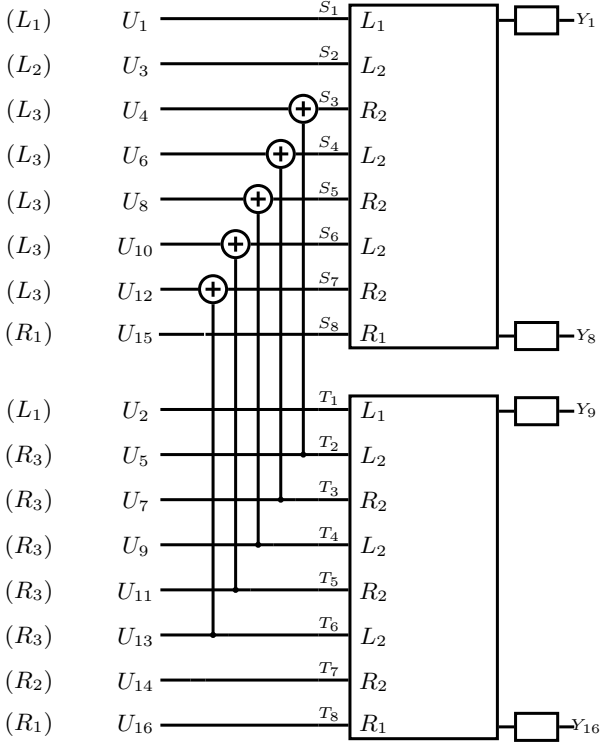}}
\caption{
\label{fig:3-level}
A $3$-level transform with $K=4$.  Each of the two blocks represents
the transform in Figure~\ref{fig:2-level}.  The channels written
inside the respective blocks correspond to $S_i\to Y_1^8S_1^{i-1}$ and
$T_i\to Y_9^{16}T_1^{i-1}$. If one labels $U_1$ to $U_{16}$ as above,
then the channels $U_i\to Y_1^{16}U_1^{i-1}$ are equivalent to the
ones on the left.
}
\end{figure}

\begin{figure}[t]
\centerline{
\scalebox{1} 
{
\begin{pspicture}(0,-6.1529684)(7.6760936,6.1929684)
\pscircle[linewidth=0.04,dimen=outer](3.15375,4.1098437){0.2}
\psline[linewidth=0.04cm](3.05375,4.1098437)(3.25375,4.1098437)
\psline[linewidth=0.04cm](3.15375,4.2098436)(3.15375,4.009844)
\psline[linewidth=0.04cm](1.401875,4.1242185)(2.95375,4.1098437)
\usefont{T1}{ptm}{m}{n}
\rput(7.320781,5.887031){\tiny $Y_1$}
\psline[linewidth=0.04cm](6.1753125,5.887031)(6.3753123,5.887031)
\psframe[linewidth=0.04,dimen=outer](6.9753127,6.087031)(6.3753123,5.687031)
\psline[linewidth=0.04cm](6.9753127,5.887031)(7.1753125,5.887031)
\psframe[linewidth=0.04,dimen=outer](6.201875,6.0442185)(3.801875,0.1242187)
\psline[linewidth=0.04cm](3.321875,4.1242185)(3.801875,4.1242185)
\psline[linewidth=0.04cm](3.801875,5.8842187)(3.481875,5.8842187)
\psline[linewidth=0.04cm](3.801875,5.0842185)(3.481875,5.0842185)
\psline[linewidth=0.04cm](3.801875,4.6042185)(3.481875,4.6042185)
\pscircle[linewidth=0.04,dimen=outer](2.83375,3.6298435){0.2}
\psline[linewidth=0.04cm](1.421875,3.6442187)(2.66,3.6529686)
\psline[linewidth=0.04cm](3.001875,3.6442187)(3.801875,3.6442187)
\pscircle[linewidth=0.04,dimen=outer](2.19375,2.0298438){0.2}
\psline[linewidth=0.04cm](1.44,2.0529685)(2.0337498,2.0498438)
\psline[linewidth=0.04cm](2.361875,2.0442188)(3.801875,2.0442188)
\pscircle[linewidth=0.04,dimen=outer](1.87375,1.5498435){0.2}
\psline[linewidth=0.04cm](1.401875,1.5642188)(1.7,1.5529686)
\psline[linewidth=0.04cm](2.041875,1.5642188)(3.801875,1.5642188)
\psline[linewidth=0.04cm](3.801875,1.0842187)(3.481875,1.0842187)
\psline[linewidth=0.04cm](3.801875,0.2842187)(3.481875,0.2842187)
\psframe[linewidth=0.04,dimen=outer](6.201875,-0.1957813)(3.801875,-6.1157813)
\psline[linewidth=0.04cm](3.801875,-0.3557813)(3.481875,-0.3557813)
\psline[linewidth=0.04cm](3.801875,-1.1557813)(3.481875,-1.1557813)
\psline[linewidth=0.04cm](3.801875,-1.6357813)(1.401875,-1.6357813)
\psline[linewidth=0.04cm](3.801875,-4.6757812)(3.481875,-4.6757812)
\psline[linewidth=0.04cm](3.801875,-5.9557815)(3.481875,-5.9557815)
\psline[linewidth=0.04cm](3.801875,-2.1157813)(1.401875,-2.1157813)
\psline[linewidth=0.04cm](3.801875,-3.7157812)(1.401875,-3.7157812)
\psline[linewidth=0.04cm](3.801875,-4.195781)(1.401875,-4.195781)
\psline[linewidth=0.04cm](3.161875,3.9242187)(3.161875,-1.6557813)
\psline[linewidth=0.04cm](2.841875,3.4442186)(2.841875,-2.1357813)
\psline[linewidth=0.04cm](2.201875,1.8642187)(2.201875,-3.7157812)
\psline[linewidth=0.04cm](1.881875,1.3842187)(1.881875,-4.195781)
\psline[linewidth=0.04cm](3.801875,-5.1557813)(3.481875,-5.1557813)
\psdots[dotsize=0.08](2.521875,3.1642187)
\psdots[dotsize=0.08](2.521875,2.8442187)
\psdots[dotsize=0.08](2.521875,2.5242188)
\psdots[dotsize=0.08](2.521875,-2.5957813)
\psdots[dotsize=0.08](2.521875,-2.9157813)
\psdots[dotsize=0.08](2.521875,-3.2357812)
\psdots[dotsize=0.06](3.641875,-5.315781)
\psdots[dotsize=0.06](3.641875,-5.4757814)
\psdots[dotsize=0.06](3.641875,-5.6357813)
\psdots[dotsize=0.06](3.641875,-0.5157813)
\psdots[dotsize=0.06](3.641875,-0.6757813)
\psdots[dotsize=0.06](3.641875,-0.8357813)
\psdots[dotsize=0.06](3.641875,5.724219)
\psdots[dotsize=0.06](3.641875,5.5642185)
\psdots[dotsize=0.06](3.641875,5.4042187)
\usefont{T1}{ptm}{m}{n}
\rput(7.420781,0.2870312){\tiny $Y_N$}
\psline[linewidth=0.04cm](6.1753125,0.2870312)(6.3753123,0.2870312)
\psframe[linewidth=0.04,dimen=outer](6.9753127,0.4870312)(6.3753123,0.0870312)
\psline[linewidth=0.04cm](6.9753127,0.2870312)(7.1753125,0.2870312)
\usefont{T1}{ptm}{m}{n}
\rput(7.550781,-0.3529688){\tiny $Y_{N+1}$}
\psline[linewidth=0.04cm](6.1753125,-0.3529688)(6.3753123,-0.3529688)
\psframe[linewidth=0.04,dimen=outer](6.9753127,-0.1529688)(6.3753123,-0.5529688)
\psline[linewidth=0.04cm](6.9753127,-0.3529688)(7.1753125,-0.3529688)
\usefont{T1}{ptm}{m}{n}
\rput(7.470781,-5.9529686){\tiny $Y_{2N}$}
\psline[linewidth=0.04cm](6.1753125,-5.9529686)(6.3753123,-5.9529686)
\psframe[linewidth=0.04,dimen=outer](6.9753127,-5.752969)(6.3753123,-6.1529684)
\psline[linewidth=0.04cm](6.9753127,-5.9529686)(7.1753125,-5.9529686)
\usefont{T1}{ptm}{m}{n}
\rput(4.084219,5.837031){\footnotesize $\B_1$}
\usefont{T1}{ptm}{m}{n}
\rput(4.244219,5.037031){\footnotesize $\B_{n-1}$}
\usefont{T1}{ptm}{m}{n}
\rput(4.084219,4.5570316){\footnotesize $\B_n$}
\usefont{T1}{ptm}{m}{n}
\rput(4.084219,4.0770316){\footnotesize $\G_n$}
\usefont{T1}{ptm}{m}{n}
\rput(4.084219,3.597031){\footnotesize $\B_n$}
\usefont{T1}{ptm}{m}{n}
\rput(4.084219,1.9970312){\footnotesize $\B_n$}
\usefont{T1}{ptm}{m}{n}
\rput(4.084219,1.5170312){\footnotesize $\G_n$}
\usefont{T1}{ptm}{m}{n}
\rput(4.244219,1.0370312){\footnotesize $\G_{n-1}$}
\usefont{T1}{ptm}{m}{n}
\rput(4.084219,0.3370312){\footnotesize $\G_1$}
\usefont{T1}{ptm}{m}{n}
\rput(0.88671875,4.0770316){\footnotesize $U_{2\ell+2}$}
\usefont{T1}{ptm}{m}{n}
\rput(0.88671875,3.597031){\footnotesize $U_{2\ell+4}$}
\usefont{T1}{ptm}{m}{n}
\rput(0.88671875,-1.6829689){\footnotesize $U_{2\ell+3}$}
\usefont{T1}{ptm}{m}{n}
\rput(0.88671875,-2.1629689){\footnotesize $U_{2\ell+5}$}
\usefont{T1}{ptm}{m}{n}
\rput(3.54625,6.037031){\scriptsize $S_1$}
\usefont{T1}{ptm}{m}{n}
\rput(3.54625,5.237031){\scriptsize $S_\ell$}
\usefont{T1}{ptm}{m}{n}
\rput(3.54625,0.4370311){\scriptsize $S_N$}
\usefont{T1}{ptm}{m}{n}
\rput(3.54625,-0.1735937){\scriptsize $T_1$}
\usefont{T1}{ptm}{m}{n}
\rput(3.54625,-0.9735937){\scriptsize $T_\ell$}
\usefont{T1}{ptm}{m}{n}
\rput(3.54625,-5.773594){\scriptsize $T_N$}
\usefont{T1}{ptm}{m}{n}
\rput(4.084219,-0.4029688){\footnotesize $\B_1$}
\usefont{T1}{ptm}{m}{n}
\rput(4.244219,-1.2029688){\footnotesize $\B_{n-1}$}
\usefont{T1}{ptm}{m}{n}
\rput(4.084219,-1.6829689){\footnotesize $\B_n$}
\usefont{T1}{ptm}{m}{n}
\rput(4.084219,-2.1629689){\footnotesize $\G_n$}
\usefont{T1}{ptm}{m}{n}
\rput(4.084219,-4.242969){\footnotesize $\B_n$}
\usefont{T1}{ptm}{m}{n}
\rput(4.084219,-3.7229688){\footnotesize $\G_n$}
\usefont{T1}{ptm}{m}{n}
\rput(4.084219,-4.7229686){\footnotesize $\G_n$}
\usefont{T1}{ptm}{m}{n}
\rput(4.244219,-5.2029686){\footnotesize $\G_{n-1}$}
\usefont{T1}{ptm}{m}{n}
\rput(4.084219,-5.902969){\footnotesize $\G_1$}
\psline[linewidth=0.04cm](2.73375,3.6298437)(2.93375,3.6298437)
\psline[linewidth=0.04cm](2.83375,3.7298436)(2.83375,3.5298438)
\psline[linewidth=0.04cm](2.09375,2.0298438)(2.29375,2.0298438)
\psline[linewidth=0.04cm](2.19375,2.1298437)(2.19375,1.9298437)
\psline[linewidth=0.04cm](1.7737501,1.5498437)(1.97375,1.5498437)
\psline[linewidth=0.04cm](1.87375,1.6498437)(1.87375,1.4498436)
\psdots[dotsize=0.09](1.881875,-4.195781)
\psdots[dotsize=0.09](2.201875,-3.7157812)
\psdots[dotsize=0.09](2.841875,-2.1157813)
\psdots[dotsize=0.09](3.161875,-1.6357813)
\end{pspicture} 
}}
\caption{
\label{fig:n-level}
The $(n+1)$-level construction.  Here, $\ell=2^{n-1}-1$.
}
\end{figure}
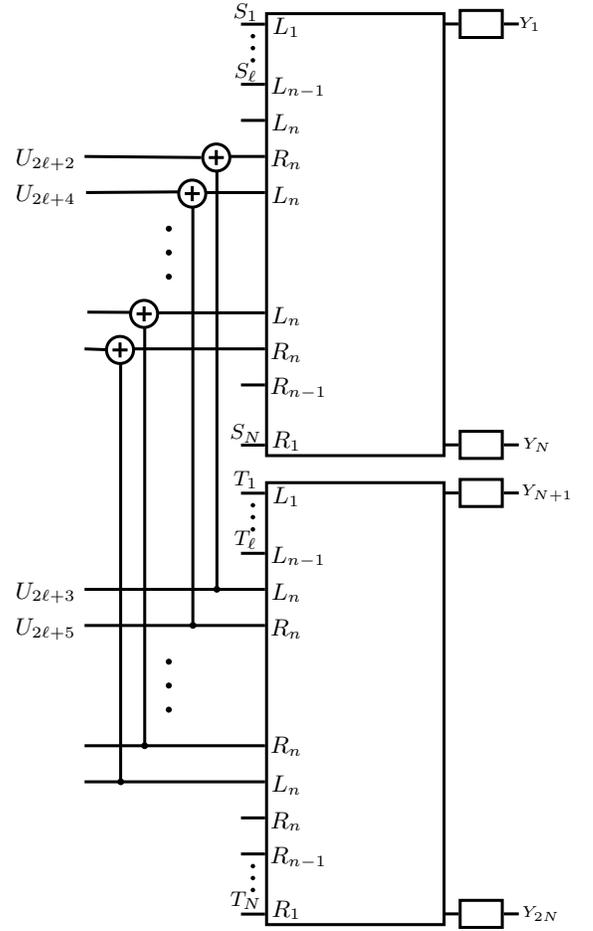

Recall that our initial goal was to ensure that all channels after the
$n$th recursion become either $\B_n$ or $\G_n$,  but the procedure
described above leaves some channels in lower levels of polarization.
The number of these less polarized channels in fact increases with
each recursion,  although fortunately, this loss is limited.  One can
indeed check that the blocklength is $N=2^{n-1}K$ after the $n$th
recursion, and the number of channels at level $1\le i\le n-1$ is
$2^{n-i}$.   Therefore the fraction of level-$n$ channels can be lower
bounded as 
\begin{align*} 
1-\frac{\sum_{i=1}^{n-1} 
	2^{n-i}}{2^{n-1}K}
	&\ge 1-\frac{2^n}{2^{n-1}K} 
	=1-\frac2K, 
\end{align*} 
which can be made arbitrarily close to~$1$ by picking a large~$K$. 

Observe that the construction described above is universal:  The
positions of the good channels (that is, $\G_n$s) after the
transformation is independent of the underlying channel $W$.  One may
therefore hope to use these channels to achieve rate $1/2$ over any
$W$ with $I(W)\ge1/2$.  Unfortunately, however, the speed of
polarization is too slow for an SC decoder to succeed. 
This is most easily seen by noting that the Bhattacharyya parameter of
$\G_{n+1}$ is given by
$$
Z(\G_{n+1})=Z(\G_n)Z(\B_n).
$$
Since $Z(\B_n)$ approaches a non-zero constant (in particular,~$1$ if
$I(W)=1/2$) as $n$ grows, the multiplicative improvement in $Z(\G_n)$
gradually slows down (to a halt if $I(W)=1/2$).  Compare this with
Ar\i kan's original method, in which identical channels are combined
in each step, and thus the improvement in the Bhattacharyya parameters
speeds up as channels become better.  The speed of universal
polarization can be bounded as follows:
\begin{proposition}
\label{prop:bounds}
Let $h\colon[0,1/2]\to[0,1]$ denote the binary entropy function, let
$a\ast b=a(1-b)+(1-a)b$, and define 
\begin{align*}
f(t,x)&=t(2x-t)\\
g(t,x)&=2x-h\big(h^{-1}(t)\ast h^{-1}(2x-t)\big),
\end{align*}
over $x\in[0,1]$ and
$t\in[\max\{0,2x-1\},x]$.  Finally define
\begin{align*}
F_0(x)&=G_0(x)=x\\
\begin{split}
F_n(x)&=f(F_{n-1}(x),x)\\
G_n(x)&=g(G_{n-1}(x),x)
\end{split}\phantom{i},\quad n=1,2,\dotsc.
\end{align*}
We have
$$
1-G_n(1-I(W))\le I(\G_n)\le 1-F_n(1-I(W)).
$$
\end{proposition}
\begin{proof}
Let $H(W)$ denote the entropy of $W$'s uniformly distributed input
conditioned on its output, that is, $H(W)=1-I(W)$.  The claim is
equivalent to
\begin{align}
\label{eqn:bounds}
F_n(H(W))\le H(\G_n)\le G_n(H(W)),
\end{align}
which holds trivially for $n=0$.  Suppose now that~\eqref{eqn:bounds}
holds for some $n\ge1$.  Recall that among all pairs of channels $V$
and $W$ with given entropies, $H((V,W)^+)$ is minimized when $V$ and
$W$ are both binary erasure channels (BECs) and maximized when both
are binary symmetric channels (BSCs)~\cite[Lemma~2.1]{Sasoglu2011}.
This implies that
\begin{align}
\label{eqn:bounds2}
f(H(\G_n),H(W))\le H(\G_{n+1})\le g(H(\G_n),H(W)).
\end{align}
On the other hand, $f(t,x)$ and $g(t,x)$ are increasing in~$t$, a
proof for the latter is given in Appendix~B.  It then follows
from~\eqref{eqn:bounds} that 
\begin{align*}
F_{n+1}(H(W))&=f(F_n(H(W)),H(W))\\
&\le f(H(\G_n),H(W))
\end{align*}
and
\begin{align*}
g(H(\G_n),H(W))&\le g(G_n(H(W)),H(W))\\
&=G_{n+1}(H(W)). 
\end{align*} 
Combining these with~\eqref{eqn:bounds2} implies the claim for $n+1$,
concluding the proof. 
\end{proof}
Observe that the above upper bound on $I(\G_n)$ is obtained by
replacing $\G_n$ and $\B_n$ with BECs with symmetric capacities
$I(\G_n)$ and $I(\B_n)$ respectively before each polarization step.
Similarly, the lower bound is obtained by replacing these channels
with BSCs.  Recall that the descendants of BECs remain BECs during
polarization, whereas those of BSCs do not remain BSCs.  This implies
that while the upper bound is achieved by the BEC, the lower bound is
loose.  Tables \ref{table:case1}~and~\ref{table:case2} list the bounds
for $I(W)=0.5$ and $I(W)=0.8$.

\begin{table}[h]
\centering
 \begin{tabular}{cll}
 \hline
  $n$ & lower bound & upper bound \\
 \hline
 0 & 0.5 &  0.5\\
 1 & 0.713 &  0.750\\
 2 & 0.771 &  0.812\\
 3 & 0.805&  0.847\\
 4 & 0.829&  0.870\\
 5 & 0.846&  0.887\\
 10 & 0.895 &  0.931\\
 20 & 0.932 &  0.960\\
 30 & 0.949 &  0.972\\
 40 & 0.958 &  0.978\\
 \hline
 \end{tabular}
 \caption{Bounds on $I(R_n)$ for $I(W) =
0.5$.}
\label{table:case1}
\end{table}

\begin{table}[h]
\centering
 \begin{tabular}{cll}
 \hline
  $n$ & lower bound & upper bound \\
 \hline
 0 & 0.8 &  0.8\\
 1 & 0.928 &  0.960\\
 2 & 0.957 &  0.986\\
 3 & 0.972 &  0.994\\
 4 & 0.981 &  0.997\\
 5 & 0.986 &  0.999\\
 10 & 0.996 &  0.999991\\
 15 & 0.9990 &  0.99999991\\
 20 & 0.9996 &  0.9999999991\\
 \hline
 \end{tabular}
 \caption{Bounds on $I(R_n)$ for $I(W) =
0.8$.}
\label{table:case2}
\end{table}
\subsection{Universal polar coding}

To obtain a good code, we can append Ar\i kan's fast (but not
universal) polarizing transform to the universal (but slow) polarizing
transform described above.  That is, once $n$ is sufficiently large so
that $I(\G_n)>1-\epsilon$ for all $W$, we may start polarizing $\G_n$
fast.  This can be done by taking $M=2^m$ copies of the slow
polarization transform and passing the $M$ copies of each $\G_n$
through the standard length-$M$ transform.  Inputs to the remaining
channels are frozen and the resulting code blocks are decoded in
succession. 

One may tailor the polar codes in the second stage to the channel that
is least degraded with respect to all channels with
$I(W)\ge1-\epsilon$.  How to find such channels is shown
in~\cite{LHU2013}.  A computationally simpler alternative is to find a
universal upper bound $Z(\G_n)\le\delta$ (as in
Proposition~\ref{prop:bounds}) and tailor the second-layer code to a
BEC with erasure probability~$\delta$.  This method is motivated by
the fact that among all channels with a fixed $Z(W)$, the BEC's
polarized descendants have the highest Bhattacharyya parameters, and
the latter can be computed in linear time~\cite{Arikan2009}.

\subsection{Rate}
Since~$I(\G_n)$ is close to~$1$, both approaches mentioned in the
paragraph above will induce a negligible rate loss in the fast
polarization stage.  Recall also that the loss in the slow
polarization stage is $O(1/K)$.  Hence the rate of the code can be
made as close to~$1/2$ as desired. 

\subsection{Error probability}
Recall that the reliabilities of the good channels after fast
polarization is $o(2^{-M^\beta})$ for all $\beta<1/2$~\cite{ArT2009},
and thus the block error probability of this code of length $NM$ is
upper bounded all $W\in\cW(1/2)$ by
$$
No(2^{-M^\beta}),
$$ 
which for fixed~$N$ vanishes as~$M$ grows.

\subsection{Complexity}

To estimate the decoding complexity, it is useful to explain the
decoding scheme in some detail:  The decoder begins by computing the
likelihood ratios for the frozen bits that see $\B_1,\dots,\B_{n-1}$,
followed by the first $\B_n$ in each block.  Then, the likelihood
ratio for the first $\G_n$ in each block is computed and passed to the
SC decoder for the first block of the fast polarization stage.  Once
this block is decoded, the bit values are passed back to the slow
polarization stage and are used to compute the likelihood ratios for
the next $\B_n$ and $\G_n$ in each block.  These steps are repeated
until all $\G_n$ blocks are decoded.  A straightforward computation
shows that the total complexity of this decoder is
$$
O(N)\kappa_{f}(M)+M\kappa_s(N),
$$
where $\kappa_f(M)$ and $\kappa_s(N)$ respectively are the decoding
complexities of the fast polar transform of length $M$ and a slow
polar transform of length $N$.  It is known~\cite{Arikan2009} that
$\kappa_f(M)=O(M\log M)$.  Now, observe that the slow polar transform
is almost identical to the fast one; it only differs in the chaining
operation in the second level and in the combination of non-identical
channels at each step.  It is easy to see that neither of these
differences affects the complexity of computing the likelihood ratios
of the polarized channels.  That is, $\kappa_s(N)=O(N\log N)$.    This
implies that the total decoding complexity at blocklength~$MN$ is
$O(MN\log MN)$, similar to regular polar codes.  Similar arguments
show that the encoding complexity is also $O(MN\log MN)$.  Note also
that the chain length~$K$ affects encoding/decoding complexities only
insofar as it appears as a linear factor in the blocklength.

\section{Codes with Arbitrary Rates}

We now discuss how to obtain universal polar codes with rates other
than $1/2$.  Recall that in the previous section we fixed the rate of
the code by using only the universally good channel $\G_n$ for coding.
When $I(\cW)$ is greater than $1/2$, the code rate can be increased by
considering coding over $\B_n$ also, since Proposition~\ref{prop:slow}
then implies $I(\B_n)>0$.  For example, when $I(\G_n)$ becomes
sufficiently close to~$1$, one may obtain more universally good
channels by only slow-polarizing $\B_n$ further.  When $I(\cW)$ is
less than $1/2$, the same method can be used by slow-polarizing $\G_n$
further once $I(\G_n)$ becomes sufficiently close to $2I(\cW)$.  Each
stage of this polarization method turns half of the remaining
nonextremal channels to extremal ones.  The resulting good channels
can then be fast-polarized for coding.  However, the blocklengths of
such constructions can be very large, since even a single stage of
slow polarization requires a large number of recursions (recall
Table~\ref{table:case1}). 

Instead, here we generalize the ideas in Section~\ref{sec:method} to
construct codes with rates $g/(b+g)$ for given positive integers
$g$~and~$b$.  Following the rate-$1/2$ case, this can be done if one
can (i)~combine $b+g$ channels at a time to create only $b+g$ channel
types after each level of slow polarization, and (ii)~ensure that~$g$
of these become better in each step and the remaining~$b$ become
worse.  As before, once the good channels become nearly perfect, one
can boost their reliabilities through fast polarization.  

It thus suffices to describe a construction that has properties
(i)~and~(ii).  Again, the simplest description is through figures.
Figure~\ref{fig:sample} shows an example of the type of transforms we
will consider.  In particular, the transform circuit consists of $b+g$
horizontal lines, each of which has a single modulo-$2$ addition that
connects it to the line below.  Starting at the second line from the
top, one can place this connection to the right or to the left of the
connection above.  
\begin{figure}[t]
\centerline{
\scalebox{1} 
{
\begin{pspicture}(0,-1.7325)(5.5490627,1.7059375)
\pscircle[linewidth=0.04,dimen=outer](3.050938,-0.3075){0.2}
\psline[linewidth=0.04cm](3.24,-0.3140625)(4.76,-0.3140625)
\psline[linewidth=0.04cm](3.050938,-0.2075)(3.050938,-0.4075)
\psline[linewidth=0.04cm](2.9509375,-0.3075)(3.1509376,-0.3075)
\psline[linewidth=0.04cm](3.050938,-0.5075)(3.050938,-0.9075)
\psline[linewidth=0.04cm](2.24,-0.9140625)(4.76,-0.9140625)
\psline[linewidth=0.04cm](1.22,-0.9140625)(1.8800001,-0.9140625)
\psline[linewidth=0.04cm](1.22,-1.5140625)(4.76,-1.5140625)
\psline[linewidth=0.04cm](1.22,0.2859375)(1.86,0.2859375)
\pscircle[linewidth=0.04,dimen=outer](2.0509374,0.2925){0.2}
\psline[linewidth=0.04cm](2.0509374,0.3925)(2.0509374,0.1925)
\psline[linewidth=0.04cm](1.9509376,0.2925)(2.1509376,0.2925)
\psline[linewidth=0.04cm](2.0509374,0.0925)(2.0509374,-0.3075)
\psline[linewidth=0.04cm](1.22,-0.3140625)(2.88,-0.3140625)
\pscircle[linewidth=0.04,dimen=outer](3.050938,0.8925){0.2}
\psline[linewidth=0.04cm](3.24,0.8859375)(4.76,0.8859375)
\psline[linewidth=0.04cm](3.050938,0.9925)(3.050938,0.7925)
\psline[linewidth=0.04cm](2.9509375,0.8925)(3.1509376,0.8925)
\psline[linewidth=0.04cm](3.050938,0.6925)(3.050938,0.2925)
\psline[linewidth=0.04cm](2.22,0.2859375)(4.76,0.2859375)
\psline[linewidth=0.04cm](1.2,0.8859375)(2.8600001,0.8859375)
\usefont{T1}{ppl}{m}{n}
{\footnotesize
\rput(0.76453125,1.5025){$\B_1^{(1)}$}
\rput(0.7445313,0.9025){$\B_1^{(2)}$}
\rput(0.78453124,0.3025){$\B_1^{(3)}$}
\rput(0.76453125,-0.2975){$\G_1^{(1)}$}
\rput(0.78453124,-0.8975){$\B_1^{(4)}$}
\rput(0.76453125,-1.4975){$\G_1^{(2)}$}
}
\pscircle[linewidth=0.04,dimen=outer](2.050938,-0.9075){0.2}
\psline[linewidth=0.04cm](2.050938,-0.8075)(2.050938,-1.0075)
\psline[linewidth=0.04cm](1.9509374,-0.9075)(2.1509376,-0.9075)
\psline[linewidth=0.04cm](2.050938,-1.1075)(2.050938,-1.5075)
\usefont{T1}{ppl}{m}{n}
\rput(5.094531,0.9025){$W$}
\usefont{T1}{ppl}{m}{n}
\rput(5.094531,0.3025){$W$}
\usefont{T1}{ppl}{m}{n}
\rput(5.094531,-0.2975){$W$}
\usefont{T1}{ppl}{m}{n}
\rput(5.094531,-0.8975){$W$}
\usefont{T1}{ppl}{m}{n}
\rput(5.094531,-1.4975){$W$}
\pscircle[linewidth=0.04,dimen=outer](4.0509377,1.4925){0.2}
\psline[linewidth=0.04cm](4.0509377,1.5925)(4.0509377,1.3925)
\psline[linewidth=0.04cm](3.9509375,1.4925)(4.1509376,1.4925)
\psline[linewidth=0.04cm](4.0509377,1.2925)(4.0509377,0.8925)
\psline[linewidth=0.04cm](1.2,1.4859375)(3.8600001,1.4859375)
\psline[linewidth=0.04cm](4.2200003,1.4859375)(4.76,1.4859375)
\usefont{T1}{ppl}{m}{n}
\rput(5.094531,1.5025){$W$}
\end{pspicture} 
}}
\caption{
\label{fig:sample}
A one-level transform that combines $b+g=6$ channels with $g=2$ and
$b=4$.  Only the labels of the channels at the corresponding locations
are shown. 
}
\end{figure}
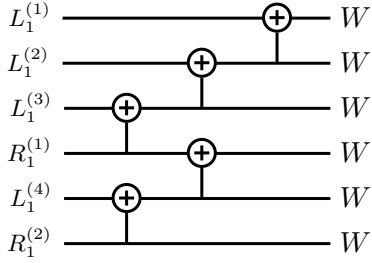

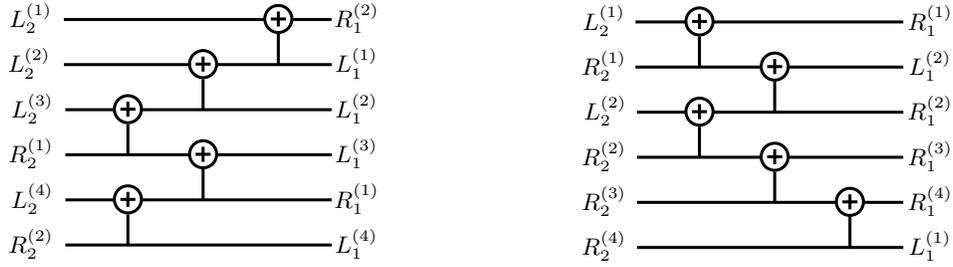
\begin{figure*}[!ht]
\centerline{
\scalebox{1} 
{
\begin{pspicture}(0,-1.7325)(5.5490627,1.7059375)
\pscircle[linewidth=0.04,dimen=outer](3.050938,-0.3075){0.2}
\psline[linewidth=0.04cm](3.24,-0.3140625)(4.76,-0.3140625)
\psline[linewidth=0.04cm](3.050938,-0.2075)(3.050938,-0.4075)
\psline[linewidth=0.04cm](2.9509375,-0.3075)(3.1509376,-0.3075)
\psline[linewidth=0.04cm](3.050938,-0.5075)(3.050938,-0.9075)
\psline[linewidth=0.04cm](2.24,-0.9140625)(4.76,-0.9140625)
\psline[linewidth=0.04cm](1.22,-0.9140625)(1.8800001,-0.9140625)
\psline[linewidth=0.04cm](1.22,-1.5140625)(4.76,-1.5140625)
\psline[linewidth=0.04cm](1.22,0.2859375)(1.86,0.2859375)
\pscircle[linewidth=0.04,dimen=outer](2.0509374,0.2925){0.2}
\psline[linewidth=0.04cm](2.0509374,0.3925)(2.0509374,0.1925)
\psline[linewidth=0.04cm](1.9509376,0.2925)(2.1509376,0.2925)
\psline[linewidth=0.04cm](2.0509374,0.0925)(2.0509374,-0.3075)
\psline[linewidth=0.04cm](1.22,-0.3140625)(2.88,-0.3140625)
\pscircle[linewidth=0.04,dimen=outer](3.050938,0.8925){0.2}
\psline[linewidth=0.04cm](3.24,0.8859375)(4.76,0.8859375)
\psline[linewidth=0.04cm](3.050938,0.9925)(3.050938,0.7925)
\psline[linewidth=0.04cm](2.9509375,0.8925)(3.1509376,0.8925)
\psline[linewidth=0.04cm](3.050938,0.6925)(3.050938,0.2925)
\psline[linewidth=0.04cm](2.22,0.2859375)(4.76,0.2859375)
\psline[linewidth=0.04cm](1.2,0.8859375)(2.8600001,0.8859375)
{\footnotesize
\usefont{T1}{ppl}{m}{n}
\rput(0.76453125,1.5025){$\B_2^{(1)}$}
\rput(0.7445313,0.9025){$\B_2^{(2)}$}
\rput(0.78453124,0.3025){$\B_2^{(3)}$}
\rput(0.76453125,-0.2975){$\G_2^{(1)}$}
\rput(0.78453124,-0.8975){$\B_2^{(4)}$}
\rput(0.76453125,-1.4975){$\G_2^{(2)}$}
}
\pscircle[linewidth=0.04,dimen=outer](2.050938,-0.9075){0.2}
\psline[linewidth=0.04cm](2.050938,-0.8075)(2.050938,-1.0075)
\psline[linewidth=0.04cm](1.9509374,-0.9075)(2.1509376,-0.9075)
\psline[linewidth=0.04cm](2.050938,-1.1075)(2.050938,-1.5075)
{\footnotesize
\usefont{T1}{ppl}{m}{n}
\rput(5.094531,1.5025){$\G_1^{(2)}$}
\rput(5.094531,0.9025){$\B_1^{(1)}$}
\rput(5.094531,0.3025){$\B_1^{(2)}$}
\rput(5.094531,-0.2975){$\B_1^{(3)}$}
\rput(5.094531,-0.8975){$\G_1^{(1)}$}
\rput(5.094531,-1.4975){$\B_1^{(4)}$}
}
\pscircle[linewidth=0.04,dimen=outer](4.0509377,1.4925){0.2}
\psline[linewidth=0.04cm](4.0509377,1.5925)(4.0509377,1.3925)
\psline[linewidth=0.04cm](3.9509375,1.4925)(4.1509376,1.4925)
\psline[linewidth=0.04cm](4.0509377,1.2925)(4.0509377,0.8925)
\psline[linewidth=0.04cm](1.2,1.4859375)(3.8600001,1.4859375)
\psline[linewidth=0.04cm](4.2200003,1.4859375)(4.76,1.4859375)
\end{pspicture} 
}

\begin{pspicture}(-2.0,-1.7164062)(5.4090624,1.6898438)
\pscircle[linewidth=0.04,dimen=outer](2.8309379,-0.32359374){0.2}
\psline[linewidth=0.04cm](3.02,-0.33015624)(4.54,-0.33015624)
\psline[linewidth=0.04cm](2.8309379,-0.22359376)(2.8309379,-0.42359376)
\psline[linewidth=0.04cm](2.7309375,-0.32359374)(2.9309375,-0.32359374)
\psline[linewidth=0.04cm](2.8309379,-0.5235937)(2.8309379,-0.92359376)
\psline[linewidth=0.04cm](4.02,-0.93015623)(4.54,-0.93015623)
\psline[linewidth=0.04cm](1.0000001,-0.93015623)(3.64,-0.93015623)
\psline[linewidth=0.04cm](1.0000001,-1.5301563)(4.54,-1.5301563)
\psline[linewidth=0.04cm](1.0000001,0.26984376)(1.6400001,0.26984376)
\pscircle[linewidth=0.04,dimen=outer](1.8309375,0.27640626){0.2}
\psline[linewidth=0.04cm](1.8309375,0.37640625)(1.8309375,0.17640625)
\psline[linewidth=0.04cm](1.7309376,0.27640626)(1.9309376,0.27640626)
\psline[linewidth=0.04cm](1.8309375,0.07640625)(1.8309375,-0.32359374)
\psline[linewidth=0.04cm](1.0000001,-0.33015624)(2.66,-0.33015624)
\pscircle[linewidth=0.04,dimen=outer](2.8309379,0.87640625){0.2}
\psline[linewidth=0.04cm](3.02,0.8698437)(4.54,0.8698437)
\psline[linewidth=0.04cm](2.8309379,0.9764063)(2.8309379,0.7764062)
\psline[linewidth=0.04cm](2.7309375,0.87640625)(2.9309375,0.87640625)
\psline[linewidth=0.04cm](2.8309379,0.67640626)(2.8309379,0.27640626)
\psline[linewidth=0.04cm](2.0,0.26984376)(4.54,0.26984376)
\psline[linewidth=0.04cm](0.9800001,0.8698437)(2.64,0.8698437)
{\footnotesize
\usefont{T1}{ppl}{m}{n}
\rput(0.56453127,1.4864062){$\B_2^{(1)}$}
\rput(0.56453127,0.88640624){$\G_2^{(1)}$}
\rput(0.58453125,0.28640625){$\B_2^{(2)}$}
\rput(0.56453127,-0.31359375){$\G_2^{(2)}$}
\rput(0.56453127,-0.91359377){$\G_2^{(3)}$}
\rput(0.56453127,-1.5135938){$\G_2^{(4)}$}

\rput(4.9,1.4864062){$\G_1^{(1)}$}
\rput(4.9,0.88640624){$\B_1^{(2)}$}
\rput(4.9,0.28640625){$\G_1^{(2)}$}
\rput(4.9,-0.31359375){$\G_1^{(3)}$}
\rput(4.9,-0.91359377){$\G_1^{(4)}$}
\rput(4.9,-1.5135938){$\B_1^{(1)}$}
}
\pscircle[linewidth=0.04,dimen=outer](3.8309379,-0.92359376){0.2}
\psline[linewidth=0.04cm](3.8309379,-0.82359374)(3.8309379,-1.0235938)
\psline[linewidth=0.04cm](3.7309375,-0.92359376)(3.9309375,-0.92359376)
\psline[linewidth=0.04cm](3.8309379,-1.1235938)(3.8309379,-1.5235938)
\pscircle[linewidth=0.04,dimen=outer](1.8309377,1.4764062){0.2}
\psline[linewidth=0.04cm](1.8309377,1.5764062)(1.8309377,1.3764062)
\psline[linewidth=0.04cm](1.7309374,1.4764062)(1.9309375,1.4764062)
\psline[linewidth=0.04cm](1.8309377,1.2764063)(1.8309377,0.87640625)
\psline[linewidth=0.04cm](0.9800001,1.4698437)(1.6600001,1.4698437)
\psline[linewidth=0.04cm](2.02,1.4698437)(4.54,1.4698437)
\end{pspicture} }
\caption{
A rate-$2/6$ transform (left) and a rate-$4/6$ transform (right).
Channels enter both transforms on the right-hand-side and produce the
channels on the left-hand-side.
\label{fig:two-rates}
}
\end{figure*}

The channels $U_i\to Y_1^{b+g}U_1^{i-1}$ produced by the transform are
defined as usual, where the inputs and outputs are numbered in
increasing order from top to bottom.  We label the channels as follows
(see Figure~\ref{fig:sample}):  If a line's connection to the bottom
is on the left side of its connection to the top, then the
corresponding channel is called an $\B$-channel.  The $i$th such
channel from the top is called $\B_1^{(i)}$.  Similarly, a channel
whose connection to the bottom is on the right side of its connection
to the top is called a $\G$-channel.  In addition, the top channel is
an $\B$-channel and the bottom channel is an $\G$-channel.  Observe
that the fraction of $\B$- and $\G$-channels can be adjusted to
arbitrary non-zero values by an appropriate choice of transform.

We will restrict our attention to two types of transforms for which
the claims will be easy to verify.  For $g\le b$ (that is, when the
target rate is less than $1/2$), we will use the transform that
produces the channels in the order
\begin{align}
\label{eqn:small-1}
\underbrace{\B\B\dots\B}_{b-g}
	\underbrace{\B\G\B\G\dots\B\G}_{g\text{ pairs}}
\end{align}
That is, the top $b-g$ channels will be of type~$\B$, followed by an
alternating sequence of $\B$- and $\G$-channels.  In order to define a
recursion, we need to specify the order in which the $b+g$ channels
enter the transform in the next level.  In the present case, the input
order is obtained by cyclically down-shifting~\eqref{eqn:small-1} by
one:
\begin{align}
\label{eqn:small-2}
\G\underbrace{\B\B\dots\B}_{b-g+1}
	\underbrace{\G\B\dots\G\B}_{g-1\text{ pairs}}
\end{align}

For the $g\ge b$ case, the top channels produced by the transform will
be of alternating types, followed by a sequence of $\G$-channels: 
\begin{align}
\label{eqn:large-1}
\underbrace{\B\G\B\G\dots\B\G}_{b\text{ pairs}}
	\underbrace{\G\G\dots\G}_{g-b}
\end{align}
These channels will be input to the next recursion after up-shifting
the order~\eqref{eqn:large-1} by one:
\begin{align}
\label{eqn:large-2}
\underbrace{\G\B\dots\G\B}_{b-1\text{ pairs}}
	\underbrace{\G\G\dots\G}_{g-b+1}\B
\end{align}
Examples of both recursions are shown in Figure~\ref{fig:two-rates}.
We will label the channels produced by these recursions as in the
previous section:  If $g\le b$, then the channels
$\B_n^{(1)},\dots,\B_n^{(b)}$ and $\G_n^{(1)}\dots\G_n^{(g)}$ after
the $n$th recursion are transformed through
\eqref{eqn:small-1}~and~\eqref{eqn:small-2} to produce
$\B_{n+1}^{(1)},\dots,\B_{n+1}^{(b)}$ and
$\G_{n+1}^{(1)}\dots\G_{n+1}^{(g)}$.  The first recursion takes $b+g$
copies of $W$ as input.  For the case $g\ge b$, the recursions are
defined through \eqref{eqn:large-1}~and~\eqref{eqn:large-2}.

The reason for the labeling above is the analogy between the
($\B$,$\G$)-channels and the channels ($\B_1$,$\G_1$) of the previous
section.  Indeed, suppose that we combine $b+g$ copies of~$W$ through
a transform that produces the channels $\B_1^{(1)},\dots,\B_1^{(b)}$
and $\G_1^{(1)},\dots,\G_1^{(g)}$.
We clearly have
$$
\sum_{i=1}^b I(\B_n^{(i)})
	+\sum_{i=1}^g I(\G_n^{(i)}) = (b+g)I(W).
$$
Moreover, the $\B$-channels are worse than $W$ and the $\G$-channels
are better:

\begin{proposition}
For all $i=1,\dots,b$ and $j=1,\dots,g$ we have
$$
I(\B_1^{(i)})\le I(W)\le I(\G_1^{(j)}).
$$
Both inequalities are strict unless $I(W)\in\{0,1\}$.
\end{proposition}
\begin{proof}
 We prove the statement for the case $g \le b$. The case $g > b$ can be analyzed
similarly. By construction, we have from top to bottom the following sequence
of channels
\[
 \underbrace{\B_1^{(1)}, \ldots, \B_1^{(b-g)}}_{b-g},
\underbrace{\B_1^{(b-g+1)}, \G_1^{(1)}, \ldots, \B_1^{(b)}, \G_1^{(g)}}_{g
\text{ pairs}}.
\]
Define $Q_1 = W$ and $Q_{i+1} = (Q_i,W)^+$ for $i = 1,2,\ldots$ If $g = 1$,
we have $\B_1^{(i)} = (Q_i,W)^-$ for $1 \le i \le b$ and $\G_1^{(1)} = (Q_b,
W)^+$. If $g > 1$, we have
\[
\B_1^{(i)} =
\begin{cases}
  (Q_i,W)^- & 1\le i \le b-g\\
  (Q_{b+g-1},W^-)^- & i=b-g+1 \\
  (W^+,W^-)^- & b-g+1 < i < b\\
  (W^+,W)^- & i=b 
\end{cases}
\]
and
\[
\G_1^{(j)} =
\begin{cases}
  (Q_{b-g+1},W^-)^+ & j=1 \\
  (W^+,W^-)^+ & 1 < j < g\\
  (W^+,W)^+ & j=g
\end{cases}.
\]
The claim then follows by noting that
\begin{align*}
 I((W,V)^-) &\le \min\{I(W),\,I(V)\} \\
 &\le \max\{I(W),\,I(V)\} \\
 &\le I((W,V)^+)
\end{align*}
for any two channels $W$ and $V$. Strict inequalities follow again
from~\cite[Lemma~2.1]{Sasoglu2011}.
\end{proof}

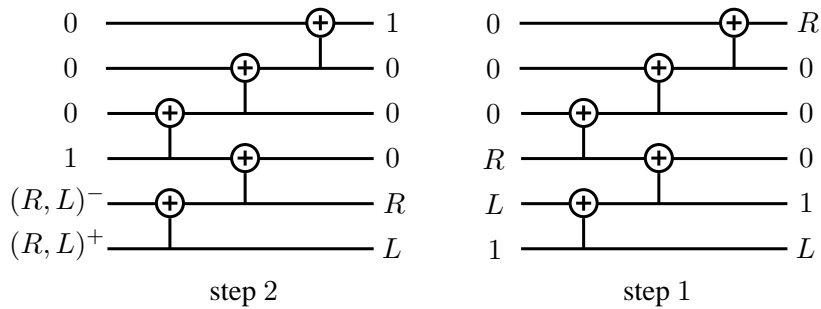
\begin{figure*}[!ht]
\centerline{
\scalebox{1} 
{
\begin{pspicture}(0,-2.0317187)(11.269062,2.0051563)
\pscircle[linewidth=0.04,dimen=outer](3.3909378,-0.00828125){0.2}
\psline[linewidth=0.04cm](3.5800002,-0.01484375)(5.1,-0.01484375)
\psline[linewidth=0.04cm](3.3909378,0.09171875)(3.3909378,-0.10828125)
\psline[linewidth=0.04cm](3.2909374,-0.00828125)(3.4909375,-0.00828125)
\psline[linewidth=0.04cm](3.3909378,-0.20828125)(3.3909378,-0.60828125)
\psline[linewidth=0.04cm](2.5800002,-0.6148437)(5.1,-0.6148437)
\psline[linewidth=0.04cm](1.5600001,-0.6148437)(2.22,-0.6148437)
\psline[linewidth=0.04cm](1.5600001,-1.2148438)(5.1,-1.2148438)
\psline[linewidth=0.04cm](1.5600001,0.58515626)(2.2,0.58515626)
\pscircle[linewidth=0.04,dimen=outer](2.3909376,0.59171873){0.2}
\psline[linewidth=0.04cm](2.3909376,0.69171876)(2.3909376,0.49171874)
\psline[linewidth=0.04cm](2.2909377,0.59171873)(2.4909377,0.59171873)
\psline[linewidth=0.04cm](2.3909376,0.39171875)(2.3909376,-0.00828125)
\psline[linewidth=0.04cm](1.5600001,-0.01484375)(3.22,-0.01484375)
\pscircle[linewidth=0.04,dimen=outer](3.3909378,1.1917187){0.2}
\psline[linewidth=0.04cm](3.5800002,1.1851562)(5.1,1.1851562)
\psline[linewidth=0.04cm](3.3909378,1.2917187)(3.3909378,1.0917188)
\psline[linewidth=0.04cm](3.2909374,1.1917187)(3.4909375,1.1917187)
\psline[linewidth=0.04cm](3.3909378,0.99171877)(3.3909378,0.59171873)
\psline[linewidth=0.04cm](2.56,0.58515626)(5.1,0.58515626)
\psline[linewidth=0.04cm](1.5400001,1.1851562)(3.2,1.1851562)
\usefont{T1}{ppl}{m}{n}
\rput(1.0745312,1.8017187){$0$}
\usefont{T1}{ppl}{m}{n}
\rput(1.0745312,0.6017187){$0$}
\usefont{T1}{ppl}{m}{n}
\rput(0.89453125,-0.5782812){$(R,L)^-$}
\usefont{T1}{ppl}{m}{n}
\rput(1.0745312,0.00171875){$1$}
\usefont{T1}{ppl}{m}{n}
\rput(1.0745312,1.2017188){$0$}
\pscircle[linewidth=0.04,dimen=outer](2.3909378,-0.60828125){0.2}
\psline[linewidth=0.04cm](2.3909378,-0.50828123)(2.3909378,-0.7082813)
\psline[linewidth=0.04cm](2.2909374,-0.60828125)(2.4909375,-0.60828125)
\psline[linewidth=0.04cm](2.3909378,-0.80828124)(2.3909378,-1.2082813)
\usefont{T1}{ppl}{m}{n}
\rput(5.3545313,1.2017188){$0$}
\usefont{T1}{ppl}{m}{n}
\rput(5.3545313,0.6017187){$0$}
\usefont{T1}{ppl}{m}{n}
\rput(5.3545313,0.00171875){$0$}
\usefont{T1}{ppl}{m}{n}
\rput(5.3745313,-0.59828126){$R$}
\usefont{T1}{ppl}{m}{n}
\rput(5.344531,-1.1982813){$L$}
\pscircle[linewidth=0.04,dimen=outer](4.390938,1.7917187){0.2}
\psline[linewidth=0.04cm](4.390938,1.8917187)(4.390938,1.6917187)
\psline[linewidth=0.04cm](4.2909374,1.7917187)(4.4909377,1.7917187)
\psline[linewidth=0.04cm](4.390938,1.5917188)(4.390938,1.1917187)
\psline[linewidth=0.04cm](1.5400001,1.7851562)(4.2000003,1.7851562)
\psline[linewidth=0.04cm](4.56,1.7851562)(5.1,1.7851562)
\usefont{T1}{ppl}{m}{n}
\rput(5.3545313,1.8017187){$1$}
\usefont{T1}{ppl}{m}{n}
\rput(0.89453125,-1.1582812){$(R,L)^+$}
\pscircle[linewidth=0.04,dimen=outer](8.890938,-0.00828125){0.2}
\psline[linewidth=0.04cm](9.08,-0.01484375)(10.6,-0.01484375)
\psline[linewidth=0.04cm](8.890938,0.09171875)(8.890938,-0.10828125)
\psline[linewidth=0.04cm](8.790937,-0.00828125)(8.990937,-0.00828125)
\psline[linewidth=0.04cm](8.890938,-0.20828125)(8.890938,-0.60828125)
\psline[linewidth=0.04cm](8.08,-0.6148437)(10.6,-0.6148437)
\psline[linewidth=0.04cm](7.06,-0.6148437)(7.7200003,-0.6148437)
\psline[linewidth=0.04cm](7.06,-1.2148438)(10.6,-1.2148438)
\psline[linewidth=0.04cm](7.06,0.58515626)(7.7000003,0.58515626)
\pscircle[linewidth=0.04,dimen=outer](7.8909373,0.59171873){0.2}
\psline[linewidth=0.04cm](7.8909373,0.69171876)(7.8909373,0.49171874)
\psline[linewidth=0.04cm](7.7909374,0.59171873)(7.9909377,0.59171873)
\psline[linewidth=0.04cm](7.8909373,0.39171875)(7.8909373,-0.00828125)
\psline[linewidth=0.04cm](7.06,-0.01484375)(8.72,-0.01484375)
\pscircle[linewidth=0.04,dimen=outer](8.890938,1.1917187){0.2}
\psline[linewidth=0.04cm](9.08,1.1851562)(10.6,1.1851562)
\psline[linewidth=0.04cm](8.890938,1.2917187)(8.890938,1.0917188)
\psline[linewidth=0.04cm](8.790937,1.1917187)(8.990937,1.1917187)
\psline[linewidth=0.04cm](8.890938,0.99171877)(8.890938,0.59171873)
\psline[linewidth=0.04cm](8.06,0.58515626)(10.6,0.58515626)
\psline[linewidth=0.04cm](7.04,1.1851562)(8.7,1.1851562)
\usefont{T1}{ppl}{m}{n}
\rput(6.6945314,1.7817187){$0$}
\usefont{T1}{ppl}{m}{n}
\rput(6.6945314,0.58171874){$0$}
\usefont{T1}{ppl}{m}{n}
\rput(6.704531,-0.61828125){$L$}
\usefont{T1}{ppl}{m}{n}
\rput(6.6945314,-0.01828125){$R$}
\usefont{T1}{ppl}{m}{n}
\rput(6.6945314,1.1817187){$0$}
\pscircle[linewidth=0.04,dimen=outer](7.890938,-0.60828125){0.2}
\psline[linewidth=0.04cm](7.890938,-0.50828123)(7.890938,-0.7082813)
\psline[linewidth=0.04cm](7.7909374,-0.60828125)(7.9909377,-0.60828125)
\psline[linewidth=0.04cm](7.890938,-0.80828124)(7.890938,-1.2082813)
\usefont{T1}{ppl}{m}{n}
\rput(10.854531,1.2017188){$0$}
\usefont{T1}{ppl}{m}{n}
\rput(10.854531,0.6017187){$0$}
\usefont{T1}{ppl}{m}{n}
\rput(10.854531,0.00171875){$0$}
\usefont{T1}{ppl}{m}{n}
\rput(10.854531,-0.59828126){$1$}
\usefont{T1}{ppl}{m}{n}
\rput(10.844531,-1.1982813){$L$}
\pscircle[linewidth=0.04,dimen=outer](9.890938,1.7917187){0.2}
\psline[linewidth=0.04cm](9.890938,1.8917187)(9.890938,1.6917187)
\psline[linewidth=0.04cm](9.790937,1.7917187)(9.990937,1.7917187)
\psline[linewidth=0.04cm](9.890938,1.5917188)(9.890938,1.1917187)
\psline[linewidth=0.04cm](7.04,1.7851562)(9.7,1.7851562)
\psline[linewidth=0.04cm](10.06,1.7851562)(10.6,1.7851562)
\usefont{T1}{ppl}{m}{n}
\rput(10.874531,1.8017187){$R$}
\usefont{T1}{ppl}{m}{n}
\rput(6.7145314,-1.2182813){$1$}
\usefont{T1}{ptm}{m}{n}
\rput(8.872188,-1.8098438){step $1$}
\usefont{T1}{ptm}{m}{n}
\rput(3.3721876,-1.8098438){step $2$}
\end{pspicture} 
}}
\caption{
\label{fig:nostuck}
An example of polarization.  Suppose that at some point in time the
channels are polarized as in the graph on the right.  Here, $0$
denotes a completely noisy channel, $1$ is a perfect channel, and $R$
and $L$ are the mediocre $R$- and $L$-channels.  Channels enter the
recursion from the right.  After the first step, the channels change
positions, but no polarization takes place.  Nevertheless, the new
positions of $R$ and $L$ ensure that they are combined in the next
step (left).  One can check that mediocre channels are always combined
eventually, regardless of their initial positions.
}
\end{figure*}

Having created $b$ bad and $g$ good channels out of $W$, we wish to
enhance polarization by making the bad channels worse and the good
channels better.   The main result of this section is that these
recursions indeed polarize channels universally:

\begin{proposition}
\hfill\phantom{x}\\
\vspace{-2ex}
\begin{itemize}
\item[(i)]
If $I(W)\ge g/(b+g)$, then for all $1\le i\le g$
$$
\lim_{n\to\infty}I(\G_n^{(i)})=1.
$$ 
\item[(ii)]
If $I(W)\le g/(b+g)$, then  for all $1\le i\le b$
$$
\lim_{n\to\infty}I(\B_n^{(i)})=0.
$$ 
\end{itemize}
\end{proposition}

\begin{proof}
We prove (i) for the case $g\le b$.  The arguments for the remaining
three cases are similar.  Recall that the recursions for the case
$g\le b$ are defined through
\eqref{eqn:small-1}~and~\eqref{eqn:small-2}. Define $Q_1 = \G_n^{(g)}$
and $ Q_{i+1} = (Q_i, \B_n^{(i)})^+$ for $1 \le i \le b-g$. When $g =
1$, we have 
$$\B_{n+1}^{(i)} = (Q_i, \B_n^{(i)})^-\text{ for }1 \le i \le b
$$
and
$$
\G_{n+1}^{(1)} = (Q_b, \B_n^{(b)})^+.
$$ 
Therefore, $I(\B_n^{(i)})$
is decreasing while in $n$ $I(\G_n^{(1)})$ is increasing. When $g >1$,
define 
\begin{align*}
P_i^+ &= (\B_n^{(b-g+i)},\G_n^{(i)})^+\\
P_i^- &= (\B_n^{(b-g+i)},\G_n^{(i)})^-
\end{align*}
for $1 \le i < g$. One can check that
the order of inputs to the recursion implies (see Figure~\ref{fig:two-rates} for
reference)
\[
\B_{n+1}^{(i)} =
\begin{cases}
  (Q_i,\B_n^{(i)})^- & 1\le i \le b-g\\
  (Q_{b-g+1},P_1^-)^- & i=b-g+1\\
  (P_{g-b+i-1}^+,P_{g-b+i}^-)^- & b-g+1 < i < b\\
  (P_{g-1}^+,\B_n^{(b)})^- & i=b
\end{cases}.
\]
Note that $I(P_{g-b+i}^-) = I((\B_n^{(i)}, \G_n^{(g-b+i)})^-) \le I(\B_n^{(i)})$
for $b-g+1 \le i < b$. It
follows that $I(\B_{n+1}^{(i)}) \le I(\B_n^{(i)})$ for all $1 \le
i \le b$. Similarly, one can check that
\[
\G_{n+1}^{(j)} =
\begin{cases}
  (Q_{b-g+1},P_1^-)^+ & j=1\\
  (P_{j-1}^+,P_{j}^-)^+ & 1 < j < g\\
  (P_{g-1}^+,\B_n^{(b)})^+ & j=g
\end{cases}.
\]
Note that 
$I(Q_{b-g+1}) \ge I(\G_n^{(g)})$ and $I(P_{j-1}^+) =
I((\B_n^{(b-g+j-1)},\G_n^{(j-1)})^+) \ge I(\G_n^{(j-1)})$ for $1 < j \le g$. It
follows that $I(\G_{n+1}^{(1)}) \ge I(\G_n^{(g)})$ and $I(\G_{n+1}^{(j)}) \ge
I(\G_n^{(j-1)})$ for $1 < j \le g$. That is the $R$-channels at level $n+1$ are
better than the ones at level $n$, with a shift in indices. 

To show that the improvement in $I(\G_n^{(i)})$ is strict unless all
the $\G$-channels are perfect, one needs to rule out the following
possibility:  If at a point in the polarization process some, but not
all, $R$-channels become perfect, then the perfect channels entering
subsequent recursions may stall the polarization of the non-perfect
ones.  We now argue that the structure of the channel combinations
does not allow this.  Suppose that all but one $R$-channels polarize
to perfect ones.  Then, there must be at least one unpolarized
$L$-channel, since otherwise the inequality $I(W)\ge g/(b+g)$ would be
violated.  Suppose that there is only one such $L$-channel $\B^{(k)}$
and all others are polarized to useless ones. One can then check that
either the unpolarized $L$- and $R$-channels will be combined in the
next recursion (which will further polarize the $R$-channel), or their
positions will change.  In particular, the $L$-channel index $k$ will
remain unchanged after each recursion, while the $R$-channel index
will be cyclically shifted by one.  If $1 \le k \le b-g+1$, then
$L^{(k)}$ will be combined with the unpolarized $R$-channel when the
$R$-channel index is shifted to $g$. On the other hand, if $b-g+1 < k
\le b$, then $L^{(k)}$ will be combined with the unpolarized
$R$-channel when the $R$-channel index is shifted to neighboring
positions $k-b+g-1$ or $k-b+g$.  Therefore, regardless of the
unpolarized $R$-channel's position, strict polarization will take
place in at most $g$ recursions.  (See Figure~\ref{fig:nostuck} for an
example of strict polarization of period two over the rate-$2/6$
recursion.)  Therefore the $R$-channel will polarize further,
eventually becoming perfect.  The same reasoning can be used when
there is more than one unpolarized $R$-channel and $L$-channel.
\end{proof}

\subsection{Polar coding}
Fix a transform of rate $g/(b+g)$.  The code construction is identical
to the one in Section~\ref{sec:method}: In the first level, channels
are combined in the usual fashion.  This is followed by a single step
of chaining~$K$ transforms that combines channels of different types.
Then, each subsequent step combines $b+g$ transform blocks in the same
fashion.  Once sufficient universal polarization is attained, the good
channels $\G_n^{(i)}$ are fast-polarized further using the Ar\i kan
transform.

\subsection{Rate}
As in the rate-$1/2$ case, the slow polarization stage involves
leaving some channels unconnected.  Similar arguments to those in
Section~\ref{sec:method} show that the fraction of unpolarized
channels is upper bounded by
$$
\frac{(g+b)^2}K,
$$
which can be made as small as desired by picking a large~$K$.

\subsection{Error probability}
Since the reliability of the good channels are determined essentially
by the fast polarization stage, the error probability of the SC decoder
can again be upper bounded for all $W\in\cW(g/b+g)$ by
$$
No(2^{-M^\beta}),
$$
where $N$ and $M$ respectively are the
lengths of the slow and the fast polarization stages. 

\subsection{Complexity}
The present construction differs from the one in
Section~\ref{sec:method} only in the size $b+g$ of the basic
transform,  and it is easily seen that the basic transforms we
discussed can be encoded and decoded in linear time.  Hence, $b+g$
does not affect the encoding and decoding complexities, which are both
$O(MN\log MN)$ for a blocklength-$MN$ code.

\section{Discussion} 
In independent work~\cite{HU2013}, Hassani and Urbanke propose two
polarization-based methods to construct universal codes.  On close
inspection, one of these methods and the one presented here are seen
to be complementary.  In particular, whereas the method here
guarantees universality in the first stage and reliability in the
second, the construction in~\cite{HU2013} reverses this order by
combining identical channels in the first stage (i.e., fast
polarization) and distinct channels in the second (i.e., slow
polarization).  It is evident from both works that many other
variations are possible for constructing universal polar codes, such
as interleaving the fast and slow polarization stages.  Such
alternatives may help reduce the impractically large blocklengths that
the present paper's methods require (see Table~\ref{table:case1}) to
simultaneously achieve universality and reliability.  For this purpose
one may also consider using larger $(b+g)$-type constructions for
simple fractional rates such as $1/2$, or mixing the unconnected
channels into the process to increase the speed of slow polarization.
The investigation of these are left for future study.

In addition to providing robustness to point-to-point channel coding,
universal polarization is also of interest from a theoretical
perspective.  Recall that one of the many appeals of polarization
methods is the ease with which they have been extended to other
communication settings.  Polar codes' optimality have already been
established for multiple-access channels~\cite{Arikan2012}, degraded
wiretap channels~\cite{MaV2011}, lossless~\cite{Arikan2010},
lossy~\cite{KoU2010}, and distributed source coding~\cite{Arikan2012}.
However, the theory is more difficult to extend to settings with two
or more receivers, and the main bottleneck appears to be the
incompatibility of polar code designs for different receivers.  With
universal polarization schemes, it may be possible to implement polar
coding to achieve the best known rates over broadcast channels,
interference channels, etc.   

It is worth mentioning that the methods discussed here also yield
universal source codes, and can be extended to non-binary alphabets
using standard arguments~\cite[Chapter~3]{Sasoglu2011}.

\section{Appendix A: Polarization preserves less noisy ordering}
Recall that designing a polar code of length $2^n$ for a channel $W$
consists in finding a set of good channels among $W^{\bf s}$, ${\bf
s}\in\{-,+\}^n$, which are defined recursively through
\begin{align*}
W^-(y_1^2|u_1)
	&=\sum_{u_2\in\{0,1\}}
	\tfrac12 W(y_1|u_1+u_2) W(y_2|u_2),\\
W^+(y_1^2,u_1|u_2)
	&=\tfrac12 W(y_1|u_1+u_2) W(y_2|u_2).
\end{align*}
A good code of rate $R<I(W)$ can be obtained by picking an $R$
fraction of these channels whose symmetric-capacities $I(W^{\bf s})$
are largest.  Here, we show that a polar code designed in this manner
for a channel is also good for all \emph{less noisy} versions of this
channels.  This result is also established independently
in~\cite{SR2013}.  Here we show it by proving the stronger statement
that the less noisy ordering of channels is preserved under
polarization.  Recall that a channel $V$ is said to be less noisy than
$W$ if $I(T;Y) \le I(T; Z)$ for all distributions of the form
\begin{align}
\label{eqn:markov}
p(t,x,y,z)=p(x,t)W(y|x)V(z|x),
\end{align}
that is, for all distributions for which $T$---$X$---$YZ$ is a Markov
chain~\cite{KM1977}.  Observe that this implies $I(W)\le I(V)$, and
thus will also imply that $I(W^{\bf s})\le I(V^{\bf s})$ for all $\bf
s$ once we show that polarization preserves the less noisy order.  Due
to the recursive description of the polarized channels, it suffices to
prove the latter claim for a single-step:

\begin{proposition}
\label{prop:ln} 
Let $W$ and $V$ be binary-input channels.  If $V$ is
less noisy than $W$, then
\begin{enumerate}
\item[(i)]
 $V^+$ is less noisy than $W^+$,
\item[(ii)]
 $V^-$ is less noisy than $W^-$. 
\end{enumerate}
\end{proposition}%

\begin{proof}
To prove (i), we need to show that $I(T;Y_1Y_2U_1) \le I(T;Z_1Z_2U_1)$
for all random variables $(T,U_1^2,Y_1^2,Z_1^2)$ that are jointly
distributed as
\begin{align}
\label{eqn:plus-markov}
\begin{split}
p(t,u_1^2&,y_1^2,z_1^2)\\
&=p(t,u_2)W^{+}(y_1^2,u_1|u_2)V^{+}(z_1^2,u_1|u_2).
\end{split}
\end{align}
Note that the channels $W^+$ and $V^+$ here share an output, namely
$U_1$, but this does not affect the mutual informations in question.
This assumption will simplify the proof.  Define $X_1=U_1+U_2$ and
$X_2=U_2$ (see Figure~\ref{fig:markov}).  We have
\begin{figure}[t]
\centerline{
\scalebox{1} 
{
\begin{pspicture}(0,-0.83921874)(5.4090624,0.83921874)
\psline[linewidth=0.04cm](1.71,0.42578125)(2.21,0.42578125)
\pscircle[linewidth=0.04,dimen=outer](2.41,0.42578125){0.2}
\psline[linewidth=0.04cm](2.61,0.42578125)(3.11,0.42578125)
\psline[linewidth=0.04cm](2.41,0.5257813)(2.41,0.32578126)
\psline[linewidth=0.04cm](2.31,0.42578125)(2.51,0.42578125)
\psline[linewidth=0.04cm](2.41,0.22578125)(2.41,-0.41421875)
\psline[linewidth=0.04cm](1.71,-0.41421875)(3.11,-0.41421875)
\psline[linewidth=0.04cm,linestyle=dotted,dotsep=0.04cm](3.73,0.42578125)(4.35,0.62578124)
\usefont{T1}{ppl}{m}{n}
\rput(4.7745314,0.6357812){$Y_1$}
\usefont{T1}{ppl}{m}{n}
\rput(4.7945313,-0.18421875){$Y_2$}
\usefont{T1}{ppl}{m}{n}
\rput(3.3845313,0.41578126){$X_1$}
\usefont{T1}{ppl}{m}{n}
\rput(3.3845313,-0.42421874){$X_2$}
\usefont{T1}{ppl}{m}{n}
\rput(1.4445312,0.39578125){$U_1$}
\usefont{T1}{ppl}{m}{n}
\rput(1.4445312,-0.44421875){$U_2$}
\psline[linewidth=0.04cm,linestyle=dotted,dotsep=0.04cm](0.53,-0.41421875)(1.03,-0.41421875)
\usefont{T1}{ppl}{m}{n}
\usefont{T1}{ppl}{m}{n}
\rput(0.27453125,-0.38421875){$T$}
\usefont{T1}{ppl}{m}{n}
\rput(4.8145313,0.21578126){$Z_1$}
\psline[linewidth=0.04cm,linestyle=dotted,dotsep=0.04cm](3.73,0.40578124)(4.35,0.20578125)
\psline[linewidth=0.04cm,linestyle=dotted,dotsep=0.04cm](3.73,-0.39421874)(4.35,-0.19421875)
\psline[linewidth=0.04cm,linestyle=dotted,dotsep=0.04cm](3.73,-0.41421875)(4.35,-0.6142188)
\usefont{T1}{ppl}{m}{n}
\rput(4.8145313,-0.6042187){$Z_2$}
\end{pspicture} 
}}
\caption{
\label{fig:markov}
Dependence graph of the random variables in~\eqref{eqn:plus-markov}.
}
\end{figure}
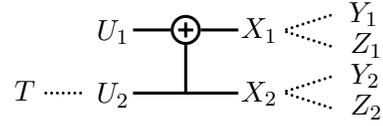
\begin{align*}
I(T;Y_1Y_2U_1)&=I(T;Y_1Y_2|U_1)\\
&= I(T;Y_1|U_1) + I(T;Y_2|Y_1U_1)\\
&\le I(T;Y_1|U_1) + I(T;Z_2|Y_1U_1)\\
&=I(T;Y_1Z_2|U_1)\\
&= I(T;Z_2|U_1) + I(T;Y_1|Z_2U_1)\\
&\le I(T;Z_2|U_1) + I(T;Z_1|Z_2U_1)\\
&=I(T;Z_1Z_2|U_1).
\end{align*}
To see the first inequality, note that 
$$
TU_1U_2X_1Y_1Z_1\text{---}X_2\text{---}Y_2Z_2
$$
is a Markov chain. Therefore we have
\begin{align*}
&p(t,u_2,x_1,x_2,y_2,z_1,z_2|y_1,u_1)\\
&\hspace{2em} = p(x_2)p(y_2,z_2|x_2)\frac{
p(t,u_1,u_2,x_1,y_1,z_1|x_2) }
	{p(y_1,u_1)}.
\end{align*}
That is, conditioned on $Y_1=y_1$ and $U_1=u_1$,
$$TU_2X_1Z_1\text{---}X_2\text{---}Y_2Z_2
$$ 
is a Markov chain, and therefore so is $T$---$X_2$---$Y_2Z_2$.  This
and the less noisiness of $V$ imply $I(T;Y_2|Y_1 = y_1, U_1 = u_1) \le
I(T;Z_2|Y_1 = y_1, U_1 = u_1)$. Averaging over $(y_1,u_1)$ yields the
first inequality.  Similarly, for the second inequality, note that
$$
TU_1U_2Y_2Z_2\text{---}X_1\text{---}Y_1Z_1
$$
is a Markov chain, and therefore so is $T$---$X_1$---$Y_1Z_1$ for
every $Z_2=z_2$ and $U_1=u_1$.  The less noisy relation then implies
$I(T;Y_1|Z_2=z_2,U_1=u_1)\le I(T;Z_1|Z_2=z_2,U_1=u_1)$.  Averaging
over $(z_2,u_1)$ yields the inequality.

To prove~(ii), we need to show that $I(T;Y_1Y_2)\le I(T; Z_1Z_2)$
for all $(T,U_1,Y_1^2,Z_1^2)$ for which
\begin{align}
\label{eqn:minus-proof}
\begin{split}
p(t,u_1,y_1^2)&=q(t,u_1)W^-(y_1^2|u_1)\\
p(t,u_1,z_1^2)&=q(t,u_1)V^-(z_1^2|u_1)
\end{split}
\end{align}
We will also define a random variable $U_2$ such that
$(T,U_1^2,Y_1^2,Z_1^2)$ is jointly distributed as
\begin{align}
\label{eqn:minus-markov}
\begin{split}
p(t,u_1^2,y_1^2,z_1^2)
&=\tfrac12 q(t,u_1)W(y_1|u_1+u_2)W(y_2|u_2)\\
	&\qquad\cdot V(z_1|u_1+u_2)V(z_2|u_2).
\end{split}
\end{align}
Observe that this definition is consistent
with~\eqref{eqn:minus-proof}, it will simplify the proof.  Defining
again $X_1=U_1+U_2$ and $X_2=U_2$ (see Figure~\ref{fig:markov2}), we
can write
\begin{align*}
I(T;Y_1Y_2) &= I(T;Y_1) + I(T;Y_2|Y_1)\\
&\le I(T;Y_1) + I(T;Z_2|Y_1)\\
&= I(T;Z_2) + I(T;Y_1|Z_2)\\
&\le I(T;Z_2) + I(T;Z_1|Z_2)\\
&= I(T;Z_1Z_2).
\end{align*}
To see the first inequality, note that the distribution
in~\eqref{eqn:minus-markov} implies that
$$
TU_1U_2X_1Y_1Z_1\text{---}X_2\text{---}Y_2Z_2
$$
is a Markov chain. Therefore we have
\begin{align*}
 &p(t,u_1,u_2,x_1,x_2,y_2,z_1,z_2|y_1)\\ 
&\hspace{2em}=p(x_2)p(y_2,z_2|x_2)\frac{p(t,u_1,u_2,x_1,y_1)}{p(y_1)}
.
\end{align*}
That is, for any fixed value of $Y_1$,
$$
TU_1U_2X_1Z_1\text{---}X_2\text{---}Y_2Z_2
$$
is a Markov chain, and therefore so is $T$---$X_2$---$Y_2Z_2$. This and
the less noisiness of $V$ imply $I(T;Y_2|Y_1 = y_1) \le I(T;Z_2|Y_1 =
y_1)$.  Averaging over $y_1$ yields the first inequality. The proof of
the second inequality follows by similar arguments.
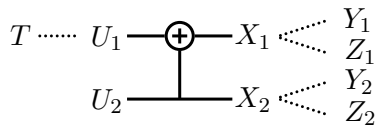
\begin{figure}[t]
\centerline{
\scalebox{1} 
{
\begin{pspicture}(0,-0.83921874)(5.4090624,0.83921874)
\psline[linewidth=0.04cm](1.71,0.42578125)(2.21,0.42578125)
\pscircle[linewidth=0.04,dimen=outer](2.41,0.42578125){0.2}
\psline[linewidth=0.04cm](2.61,0.42578125)(3.11,0.42578125)
\psline[linewidth=0.04cm](2.41,0.5257813)(2.41,0.32578126)
\psline[linewidth=0.04cm](2.31,0.42578125)(2.51,0.42578125)
\psline[linewidth=0.04cm](2.41,0.22578125)(2.41,-0.41421875)
\psline[linewidth=0.04cm](1.71,-0.41421875)(3.11,-0.41421875)
\psline[linewidth=0.04cm,linestyle=dotted,dotsep=0.04cm](3.73,0.42578125)(4.35,0.62578124)
\usefont{T1}{ppl}{m}{n}
\rput(4.7745314,0.6357812){$Y_1$}
\usefont{T1}{ppl}{m}{n}
\rput(4.7945313,-0.18421875){$Y_2$}
\usefont{T1}{ppl}{m}{n}
\rput(3.3845313,0.41578126){$X_1$}
\usefont{T1}{ppl}{m}{n}
\rput(3.3845313,-0.42421874){$X_2$}
\usefont{T1}{ppl}{m}{n}
\rput(1.4445312,0.39578125){$U_1$}
\usefont{T1}{ppl}{m}{n}
\rput(1.4445312,-0.44421875){$U_2$}
\psline[linewidth=0.04cm,linestyle=dotted,dotsep=0.04cm](0.53,0.42578125)(1.03,0.42578125)
\usefont{T1}{ppl}{m}{n}
\rput(0.29453126,0.43578124){$T$}
\usefont{T1}{ppl}{m}{n}
\usefont{T1}{ppl}{m}{n}
\rput(4.8145313,0.21578126){$Z_1$}
\psline[linewidth=0.04cm,linestyle=dotted,dotsep=0.04cm](3.73,0.40578124)(4.35,0.20578125)
\psline[linewidth=0.04cm,linestyle=dotted,dotsep=0.04cm](3.73,-0.39421874)(4.35,-0.19421875)
\psline[linewidth=0.04cm,linestyle=dotted,dotsep=0.04cm](3.73,-0.41421875)(4.35,-0.6142188)
\usefont{T1}{ppl}{m}{n}
\rput(4.8145313,-0.6042187){$Z_2$}
\end{pspicture} 
}}
\caption{
\label{fig:markov2}
Dependence graph of the random variables in~\eqref{eqn:minus-markov}.
}
\end{figure}
\end{proof}

Note that the choice of the polarization transform and the alphabet
size are immaterial to the proof above, and thus the result holds in
more generality as long as the polarized channels are appropriately
defined. 

It is of interest to characterize the weakest relations that are
preserved under polarization.  For example, the \emph{more capable}
relation~\cite{KM1977}, which is weaker than the less noisy relation,
is not preserved.  A channel $V(z|x)$ is said to be more capable than
$W(y|x)$ if $I(X;Y) \le I(X;Z)$ for all $p(x,y,z)=p(x)W(y|x)V(z|x)$.
It is shown in~\cite[Lemma~7.1]{Sasoglu2011} that in the class of
symmetric binary-input channels with a given capacity, the binary
symmetric channel is the least capable.  However, the `minus' version
of the binary symmetric channel has a larger capacity than the minus
versions of all other channels in this class, and therefore is not
less capable than any such channel.  See, for
example,~\cite[Lemma~2.1]{Sasoglu2011} for a proof.

\section{Appendix B}
\begin{lemma}

For every $x\in[0,1]$, the function
$$
g_x(t)=h(h^{-1}(t)*h^{-1}(2x-t))
$$
defined over $t\in[\max\{0,2x-1\},x]$ is decreasing.
\end{lemma}

\begin{proof} The function $h$ is monotonically increasing over $[0,1/2]$. So
it suffices to show $k_x(t) := h^{-1}(t) * h^{-1}(2x-t)$ is decreasing in $t$.
Defining $f=h^{-1}$, some algebra yields
\begin{align*}
&\tfrac{d}{dt}k_x(t)\\
	&=f'(t)[1-2f(2x-t)] 
	- f'(2x-t)[1-2f(t)].
\end{align*}
The right-hand-side of the above is non-positive since
$$
f'(t) \le f'(2x-t) \text{ and } f(t) \le f(2x-t)\le1/2,
$$
which follow from $f$ being convex, increasing, and $[0,1/2]$-valued.
%
\end{proof}

\section*{Acknowledgement}
We thank S.~H.~Hassani and R.~Urbanke for stimulating discussions.
E.~\c Sa\c so\u glu's work is supported by the Swiss National Science
Foundation under grant PBELP2\_137726.  L.~Wang's work is funded by
the 2013 Qualcomm Innovation Fellowship.
\bibliographystyle{IEEEtran}
\bibliography{univ}

\newcommand{\noopsort}[1]{}
\begin{thebibliography}{10}
\providecommand{\url}[1]{#1}
\csname url@samestyle\endcsname
\providecommand{\newblock}{\relax}
\providecommand{\bibinfo}[2]{#2}
\providecommand{\BIBentrySTDinterwordspacing}{\spaceskip=0pt\relax}
\providecommand{\BIBentryALTinterwordstretchfactor}{4}
\providecommand{\BIBentryALTinterwordspacing}{\spaceskip=\fontdimen2\font plus
\BIBentryALTinterwordstretchfactor\fontdimen3\font minus
  \fontdimen4\font\relax}
\providecommand{\BIBforeignlanguage}[2]{{%
\expandafter\ifx\csname l@#1\endcsname\relax
\typeout{** WARNING: IEEEtran.bst: No hyphenation pattern has been}%
\typeout{** loaded for the language `#1'. Using the pattern for}%
\typeout{** the default language instead.}%
\else
\language=\csname l@#1\endcsname
\fi
#2}}
\providecommand{\BIBdecl}{\relax}
\BIBdecl

\bibitem{BBT1959}
D.~Blackwell, L.~Breiman, and A.~J. Thomasian, ``The capacity of a class of
  channels,'' \emph{Ann. Math. Statist.}, vol.~30, no.~4, pp. 1229--1241, 1959.

\bibitem{KRU2012}
\BIBentryALTinterwordspacing
S.~Kudekar, T.~Richardson, and R.~L. Urbanke, ``Spatially coupled ensembles
  universally achieve capacity under belief propagation,'' 2012. [Online].
  Available: \url{http://arxiv.org/abs/1201.2999}
\BIBentrySTDinterwordspacing

\bibitem{Arikan2009}
E.~Ar{\i}kan, ``Channel polarization: {A} method for constructing
  capacity-achieving codes for symmetric binary-input memoryless channels,''
  \emph{{IEEE} Trans. Inf. Theory}, vol.~55, no.~7, pp. 3051--3073, Jul. 2009.

\bibitem{HKU2009}
S.~Hassani, S.~Korada, and R.~L. Urbanke, ``The compound capacity of polar
  codes,'' in \emph{Proc. 47th Ann. Allerton Conf. Comm. Control Comput.},
  Monticello, IL, Sep. 2009, pp. 16--21.

\bibitem{Sasoglu2011}
E.~{\c Sa\c so\u glu}, ``Polar coding theorems for discrete systems,'' {Ph.D.}
  Thesis, Ecole Polytechnique F{\'e}d{\'e}rale de Lausanne, Switzerland, Sep.
  2011.

\bibitem{Korada2009}
S.~B. Korada, ``Polar codes for channel and source coding,'' {Ph.D.} Thesis,
  Ecole Polytechnique F{\'e}d{\'e}rale de Lausanne, Switzerland, Sep. 2009.

\bibitem{LHU2013}
\BIBentryALTinterwordspacing
W.~Liu, S.~H. Hassani, and R.~L. Urbanke, ``The least degraded and the least
  upgraded channel with respect to a channel family,'' 2013. [Online].
  Available: \url{http://arxiv.org/abs/1304.5150}
\BIBentrySTDinterwordspacing

\bibitem{ArT2009}
E.~Ar{\i}kan and E.~Telatar, ``On the rate of channel polarization,'' in
  \emph{Proc. {IEEE} Int. Symp. Inf. Theory}, Seoul, South Korea, Jul. 2009,
  pp. 1493--1495.

\bibitem{HU2013}
\BIBentryALTinterwordspacing
S.~H. Hassani and R.~L. Urbanke, ``Universal polar codes,'' 2013. [Online].
  Available: \url{http://arxiv.org/abs/1307.7223}
\BIBentrySTDinterwordspacing

\bibitem{Arikan2012}
E.~Ar{\i}kan, ``Polar coding for the slepian-wolf problem based on monotone
  chain rules,'' in \emph{Proc. {IEEE} Int. Symp. Inf. Theory}, Cambridge, MA,
  2012, pp. 566--570.

\bibitem{MaV2011}
H.~Mahdavifar and A.~Vardy, ``Achieving the secrecy capacity of wiretap
  channels using polar codes,'' \emph{{IEEE} Trans. Inf. Theory}, vol.~57,
  no.~10, pp. 6428--6443, Oct. 2011.

\bibitem{Arikan2010}
E.~Ar{\i}kan, ``Source polarization,'' in \emph{Proc. {IEEE} Int. Symp. Inf.
  Theory}, Austin, TX, Jun. 2010, pp. 899--903.

\bibitem{KoU2010}
S.~Korada and R.~Urbanke, ``Polar codes are optimal for lossy source coding,''
  \emph{{IEEE} Trans. Inf. Theory}, vol.~56, no.~4, pp. 1751--1768, Apr. 2010.

\bibitem{SR2013}
\BIBentryALTinterwordspacing
D.~Sutter and J.~M. Renes, ``Universal polar codes for more capable and less
  noisy channels and sources,'' 2013. [Online]. Available:
  \url{http://arxiv.org/abs/1312.5990}
\BIBentrySTDinterwordspacing

\bibitem{KM1977}
J.~K{\"o}rner and K.~Marton, ``Comparison of two noisy channels,'' in
  \emph{Topics in Information Theory (Colloquia Mathematica Societatis J\'anos
  Bolyai, Keszthely, Hungary, 1975)}, I.~Csisz{\'a}r and P.~Elias, Eds.\hskip
  1em plus 0.5em minus 0.4em\relax Amsterdam: North-Holland, 1977, pp.
  411--423.

\end{thebibliography}


\newcommand{\noopsort}[1]{}

\end{document}